\documentclass{amsart}

\usepackage{hyperref}
\usepackage{amsmath}
\usepackage{amssymb}
\usepackage{mathrsfs}
\usepackage{amscd}
\usepackage{graphicx}
\usepackage{mathdots}
\usepackage{gensymb}
\usepackage{epstopdf}
 \usepackage{todonotes}
 \usepackage{youngtab}
 \usepackage{wrapfig}
 \usepackage{bbm}

\usepackage{epsfig}       %For postscript
\usepackage{epic,eepic}        %For epic and eepic output from xfig

\newtheorem{thm}{Theorem}[section]
\newtheorem{prop}[thm]{Proposition}
\newtheorem{lem}[thm]{Lemma}
\newtheorem{cor}[thm]{Corollary}
\newtheorem{conj}[thm]{Conjecture}

\theoremstyle{definition}

\theoremstyle{remark}
\newtheorem{remark}[thm]{Remark}

\numberwithin{equation}{section}

\newcommand{\Tr}{\operatorname{Tr}}

\newcommand{\C}{\mathbb{C}}
\newcommand{\R}{\mathbb{R}}

\newcommand{\gog}{\mathfrak{g}}
\newcommand{\goh}{\mathfrak{h}}
\newcommand{\tot}{\mathfrak{t}}

\def\SU{{\rm SU}}\def\U{{\rm U}}

\begin{document}

\title{Projections of Orbital Measures and Quantum Marginal Problems}
%\keywords{}
%\subjclass[2010]{test}

\author{Beno\^it Collins}
\address{Department of Mathematics, Kyoto University}
\email{collins@math.kyoto-u.ac.jp}

 \author{Colin McSwiggen}
\address{Courant Institute of Mathematical Sciences, New York University}
\email{csm482@nyu.edu}

 \maketitle

\begin{abstract}
This paper studies projections of uniform random elements of (co)adjoint orbits of compact Lie groups.  Such projections generalize several widely studied ensembles in random matrix theory, including the randomized Horn's problem, the randomized Schur's problem, and the orbital corners process.  In this general setting, we prove integral formulae for the probability densities, establish some properties of the densities, and discuss connections to multiplicity problems in representation theory as well as to known results in the symplectic geometry literature.  As applications, we show a number of results on marginal problems in quantum information theory and also prove an integral formula for restriction multiplicities.
\end{abstract}

 \tableofcontents

\section{Introduction}

Recent years have seen a surge of interest among random matrix theorists in projections of random elements of Lie group orbits.  The following three problems, for example, have attracted particular attention.  Let $A, B$ be two fixed $n$-by-$n$ real diagonal matrices, and let $U \in \U(n)$ be a Haar-distributed random unitary matrix.

\begin{itemize}
\item \textbf{The randomized Horn's problem \cite{Z, CZ1, CZ2, CMZ, CMZ2, McS-splines, BGH, ForresterZhang, ZKF}:} What is the joint distribution of the eigenvalues of $A + UBU^\dagger$?
\item \textbf{The randomized Schur's problem \cite{CMZ2, CZ-SchurKostka, BGH}:} What is the joint distribution of the diagonal entries of $U B U^\dagger$?
\item \textbf{The orbital corners process \cite{Barysh, CC-corners, Z2}:} For $1 \le k < n$, what is the joint distribution of the eigenvalues of the upper left $k$-by-$k$ submatrix of $U B U^\dagger$?
\end{itemize}

Although the above problems may not appear to be obviously related, they can all be put in a common form.  Suppose we are given a connected, compact Lie group $G$ and a closed, connected subgroup $H \subset G$.  Write $\gog$, $\goh$ for the respective Lie algebras of these two groups, choose an $\mathrm{Ad}$-invariant inner product on $\gog$, and let $\phi: \gog \to \goh$ be the orthogonal projection from $\gog$ onto $\goh$.  Fix a Cartan subalgebra $\tot \subset \goh$ and a dominant Weyl chamber $\tot_+ \subset \tot$, and let $\psi : \goh \to \tot_+$ be the projection along orbits, i.e.~the map that sends each adjoint $H$-orbit in $\goh$ to its unique representative in $\tot_+$.  Finally, given $\lambda \in \gog$, write $\mathcal{O}_\lambda = \{ \mathrm{Ad}_g \lambda \ | \ g \in G \}$ for its orbit under the adjoint action of $G$.  The orbit $\mathcal{O}_\lambda$ carries a unique $G$-invariant probability measure, called the {\it orbital measure}, obtained as the pushforward of the Haar probability measure on $G$.

We now pose the following question.  Let $X$ be a random element of $\mathcal{O}_\lambda$ distributed according to the orbital measure, that is, $X = \mathrm{Ad}_g \lambda$ where $g$ is a Haar-distributed random element of $G$.  What is the distribution of $\psi(\phi(X)) \in \tot_+$?  If $G = \U(n) \times \U(n)$ and $H \cong \U(n)$ is the diagonal subgroup, this question reduces to the randomized Horn's problem; if $G = \U(n)$ and $H$ is a maximal torus, it reduces to the randomized Schur's problem; and if $G = \U(n)$ with
$$H = \{ V \oplus \mathbbm{1}_{n-k} \ | \ V \in \U(k) \} \cong \U(k),$$
where $\mathbbm{1}_{n-k}$ is the $(n-k)$-by-$(n-k)$ identity matrix, then it reduces to the orbital corners process.  For details, see the papers cited above.

In the present paper, we study the general case where $G$ and $H$ are arbitrary up to the assumptions of compactness and connectedness.  We derive an integral formula for the distribution of $\psi(\phi(X))$, use this formula to study properties of this distribution, and discuss connections to multiplicity problems in combinatorial representation theory.  These results unify and generalize a number of statements that were previously known to hold in the special cases described above.  Our representation-theoretic results include an integral formula that expresses restriction multiplicities in terms of the spectral density of a random matrix model.  We also develop detailed applications to quantum marginal problems, a subject in quantum information theory that has been the focus of considerable research by both mathematicians and physicists over the past two decades.  We provide background on quantum marginal problems in the following subsection.

The problem of describing the distribution of $\psi(\phi(X))$ is well known in symplectic geometry.  Although we make minimal direct use of symplectic geometry in this paper, we now briefly explain how the problem arises in that context.  Since $G$ is compact, we can use an $\mathrm{Ad}$-invariant inner product to identify $\gog$ with its dual $\gog^*$, giving an isometry between adjoint and coadjoint orbits of $G$.  Coadjoint orbits of Lie groups are well known to carry a canonical symplectic form (see e.g.~\cite{Kirillov-lectures}), for which the action of $H$ is Hamiltonian and the projection $\phi$ is a moment map.  Furthermore, the invariant measure on the orbit $\mathcal{O}_\lambda$ coincides up to normalization with the Liouville measure associated to this symplectic form.  Thus the distribution of $\phi(X)$ is equal, up to a constant, to the Duistermaat--Heckman measure for the action of $H$ on $\mathcal{O}_\lambda$.  Duistermaat--Heckman measures, in general, have been the object of intense study by symplectic geometers since the seminal works \cite{DH, Heck}, and most of what is currently known about the problems studied in this paper comes from symplectic geometry.  However, here we take a different approach, based primarily on harmonic analysis.  Our methods build on recent joint and independent work by Coquereaux, Zuber, and the second author, which applied similar techniques to study the randomized Horn's problem, the randomized Schur's problem, and the orbital corners process \cite{Coq2, CMZ, CMZ2, CZ1, CZ-SchurKostka, McS-splines, Z, Z2}.

\subsection{Quantum marginal problems}
The main applications in this paper concern so-called {\it quantum marginal problems} (also known as $N$-representability problems in quantum chemistry \cite{Ruskai}), which have recently been an important subject of research in quantum information theory.  The most basic example of a quantum marginal problem concerns two Hilbert spaces $V$ and $W$, which we regard as the state spaces of two distinct quantum-mechanical systems.  Given a pair of vectors $v \in V$, $w \in W$, we would like to know whether these are compatible as reduced states, or {\it marginals}, of a given joint state in the composite system with state space $V \otimes W$.

In the case where $V \cong \C^m$ and $W \cong \C^n$ are both finite dimensional, a complete solution was given in 2004 by Klyachko \cite{Klyachko-marginals}, who derived a system of linear inequalities on $v$ and $w$ that fully characterize the possible marginals for any given joint state in $V \otimes W$.  This finite-dimensional problem can be stated as follows.

An $mn$-by-$mn$ Hermitian matrix $M$ can be regarded as an operator on $V \otimes W \cong \C^{mn}$. We then write $M_{ij,kl}$ for the entries of $M$ relative to the basis $e_i \otimes f_j$, where $\{e_i \}_{i=1}^m$ and $\{ f_j \}_{j=1}^n$ are the standard bases of $\C^m$ and $\C^n$ respectively, and the double indices $ij$, $kl$ run from $11$ to $mn$ in lexicographic order.  The {\it marginals} of $M$ are the $m$-by-$m$ Hermitian matrix $\pi_1(M)$ and the $n$-by-$n$ Hermitian matrix $\pi_2(M)$ obtained by taking the partial trace over each subsystem.  That is, the entries of the marginals are given by
\begin{align}
\begin{split} \label{eqn:marginal-def}
\pi_1(M)_{ik} &= \sum_{j=1}^{n} M_{ij,kj}, \qquad 1 \le i, k \le m,\\
 \pi_2(M)_{jl} &= \sum_{i=1}^{m} M_{ij,il}, \qquad 1 \le j, l \le n.
 \end{split}
\end{align}
In concrete mathematical terms, Klyachko answered the following question.  Suppose we are given vectors $\lambda \in \R^{mn}$, $\mu \in \R^m$, and $\nu \in \R^n$, with coordinates in non-increasing order (i.e.~$\lambda_{11} \ge \hdots \ge \lambda_{mn}$, and likewise for $\mu$ and $\nu$).  Does there exist an $mn$-by-$mn$ Hermitian matrix $M$ with eigenvalues $(\lambda_{11}, \hdots, \lambda_{mn})$, such that $\pi_1(M)$ has eigenvalues $(\mu_1, \hdots, \mu_m)$ and $\pi_2(M)$ has eigenvalues $(\nu_1, \hdots, \nu_n)$?  In particular, for any given $\lambda$, Klyachko gave a recursive procedure for enumerating a list of linear inequalities that cut out a polytope in $\R^{m+n}$ whose points are exactly the compatible pairs $(\mu, \nu)$.  For a mathematical review of this problem and its solution, see \cite{Knut-lectures}.  For reviews from a physics perspective, see \cite{Schilling-review, Tyc-Vlach}.

The problem just described corresponds to the case of two distinguishable particles, but we can pose similar questions about bosons, fermions, or particles that obey more exotic statistics.  For systems of $k$ indistinguishable bosons or fermions, states correspond to self-adjoint operators acting on $\mathrm{Sym}^k \C^n$ or $\wedge^k \C^n$ respectively; we can regard such operators as $n^k$-by-$n^k$ Hermitian matrices that satisfy a symmetry or antisymmetry constraint.  The quantum marginal problem in this setting asks: given $\lambda \in \R^{n^k}$ and $\mu \in \R^n$, does there exist a Hermitian matrix $M$ with eigenvalues $\lambda$, {\it satisfying the (anti)symmetry constraint}, whose single-particle marginals have eigenvalues $\mu$?

In this paper, we study extensions of the above quantum marginal problems in terms of random matrices.  In the case of two distinguishable particles, we ask: given a uniform random $mn$-by-$mn$ Hermitian matrix $M$ with spectrum $\lambda$, what is the joint probability distribution of the spectra of $\pi_1(M)$ and $\pi_2(M)$?  This question can be regarded as a quantitative refinement of the original quantum marginal problem.  Whereas the original problem only asks whether any $M$ with the given marginal spectra $(\mu, \nu)$ exists, here we want to know -- in a loose sense -- how many such $M$ there are.  That is, we now want a complete description of the distribution of $(\mu, \nu)$, while the original problem only asks about the support of this distribution.  Similarly, for systems of bosons or fermions, we ask: given a uniform random self-adjoint operator on $\mathrm{Sym}^k \C^n$ or $\wedge^k \C^n$ with fixed eigenvalues, what is the probability distribution of the spectrum of a single-particle marginal?  Each of these problems can be framed in terms of projections of adjoint orbits of a compact Lie group.

This randomized version of the problem is a natural question from the point of view of quantum statistical mechanics, where numerous authors have investigated the properties of random quantum states that are ``typical'' in the sense of being uniformly distributed conditional on certain information \cite{LloydPagels, SZ1, SZ2, SZ3, BHKRV}.  The systematic study of randomized quantum marginal problems was initiated by Christandl, Doran, Kousidis and Walter \cite{CDKW} using techniques from symplectic geometry.  In particular, they gave an efficient algorithm for computing the spectral distributions of quantum marginals for distinguishable particles, bosons, and fermions.  While this represents a solution to the above problems in a computational sense, our understanding of this subject is still far from complete, and a number of unanswered questions remain.  For example, from the perspective of random matrix theory, one would like to understand the asymptotic behavior of the spectral distributions as the dimension of the system or the number of particles tends to infinity.  However, the known methods for computing the distributions do not appear to yield concrete formulae that could be used to approach such questions analytically.  Accordingly, here we take a different, complementary approach to \cite{CDKW}, relying mainly on harmonic analysis.  A further advantage of the analytic approach is that it can be used to study projections of group orbits in more general settings where the orbits do not carry a symplectic structure; see e.g.~\cite{CZ2}.  In \cite{MatsMcS}, Matsumoto and the second author have developed a third distinct approach to randomized quantum marginal problems based on moment methods and Weingarten calculus.

The contributions of the present work include integral formulae for the spectral distributions for the randomized quantum marginal problems for distinguishable particles, indistinguishable bosons, and indistinguishable fermions, as well as results on basic properties of these distributions, such as the degree of differentiability of their densities.  We also relate the probability densities for distinguishable particles to Kronecker coefficients, i.e.~the tensor product multiplicities for irreducible representations of the symmetric group.  It is well known that these probability densities are asymptotically related to Kronecker coefficients in the semiclassical limit, but here we record an exact (non-asymptotic) relation, which follows from a box spline convolution identity due to De Concini, Procesi and Vergne \cite{DPV}.

Finally, at the end of the paper we identify some questions for further research.  Recently Belinschi, Guionnet and Huang proved large-deviations principles for the randomized Horn's and Schur's problems in the large-rank limit, and used these to derive related large-deviations principles for Kostka numbers and Littlewood--Richardson coefficients \cite{BGH}.  Our hope is that the present work will be a step toward a similar analysis of the asymptotic behavior of quantum marginal problems and of Kronecker coefficients.  We also discuss possible applications to quantum de Finetti theorems and tests of entanglement for quantum states.

\subsection{Organization of the paper}
In Section \ref{sec:formulae}, we derive general integral formulae for projections of orbital measures of a connected, compact Lie group onto the Lie algebra of a closed, connected subgroup.  We then consider the special cases of the two-body and multi-body finite-dimensional quantum marginal problems for distinguishable particles, as well as for indistinguishable bosons and fermions.  The proofs are essentially calculations in Fourier analysis, which make extensive use of celebrated formulae of Harish-Chandra and Itzykson--Zuber.  In Section \ref{sec:properties}, we prove some basic properties of these measures' probability density functions.  It is an immediate consequence of the Duistermaat--Heckman theorem and Kirwan convexity theorem in symplectic geometry that such a density is a piecewise polynomial function supported on a convex polytope; here we prove a lower bound on the number of continuous derivatives of the density at a point where it is not analytic, and we compute the maximum degree of the local polynomial expressions in the case of quantum marginal problems.  In Section \ref{sec:multiplicities}, we discuss representation-theoretic restriction multiplicities, which are related to projections of orbital measures by a semiclassical limit in the sense of geometric quantization as elaborated by Guillemin--Sternberg \cite{GS}.  In the case of the two-body quantum marginal problem, these multiplicities include the well-known Kronecker coefficients.  Finally, in Section \ref{sec:further}, we discuss some questions for further research, including the possibility of proving large-deviations principles for quantum marginal problems and Kronecker coefficients.

\section{Integral formulae for the densities} \label{sec:formulae}

In this section, we prove integral formulae for the joint spectral distributions of random matrix models related to projections of coadjoint orbits.  In Section \ref{sec:projections-general} we consider the general case of a connected, compact Lie group $G$ and a closed, connected subgroup $H \subset G$.  We then turn to quantum marginal problems, treating separately the cases of distinguishable particles (Section \ref{sec:formula-dist}), indistinguishable bosons (Section \ref{sec:formula-bos}) and indistinguishable fermions (Section \ref{sec:formula-fer}).  In what follows, we assume familiarity with basic notions of Lie theory.  We refer the reader to \cite{Hall} or \cite{Helgason-DLS} for introductions to Lie groups and Lie algebras, and to \cite[\textsection1.3]{McS-HC} for a brief review of definitions.

\subsection{Projections of orbital measures} \label{sec:projections-general}
Let $G$ be a connected, compact Lie group, $H \subset G$ a closed, connected subgroup, and $\gog, \goh$ their respective Lie algebras.  Let $\langle \cdot, \cdot \rangle$ be an $\mathrm{Ad}$-invariant inner product on $\gog$, which we use to identify $\gog \cong \gog^*$ and $\goh \cong \goh^*$, and let $\phi: \gog \to \goh$ be the orthogonal projection.  Write $\mathcal{O}_\lambda \subset \gog$ for the (co)adjoint orbit of $\lambda \in \gog$.  This orbit carries a unique $G$-invariant probability measure, the {\it orbital measure} $\beta_\lambda$.

Given a random element $X \in \mathcal{O}_\lambda$ with distribution $\beta_\lambda$, its projection $\phi(X) \in \goh$ lies in a random (co)adjoint orbit of $H$.  We want to study the resulting distribution on the space of $H$-orbits.

To make this problem concrete, let $\tot \subset \goh$ be a Cartan subalgebra and $\Phi_\goh \subset \tot^*$ the roots of $\goh$ with respect to $\tot$, which we identify with elements of $\tot$ via the inner product.\footnote{Here we take the roots to be real-valued linear functionals on $\tot$, so that their definition differs by a factor of $i$ from the convention usually used in the study of complex Lie algebras.}  Fix a choice $\Phi_\goh^+$ of positive roots, let $\tot_+ \subset \tot$ be the positive Weyl chamber, and let $\psi : \goh \to \tot_+$ be the map that sends each element of $\goh$ to the unique representative of its $H$-orbit in $\tot_+$.  We will derive an integral formula for the pushforward measure $\psi_* \phi_* \beta_\lambda$.

Before stating the formula, we must introduce some more notation.  Let $\tilde \tot \supset \tot$ be a Cartan subalgebra of $\gog$ containing $\tot$, $\Phi_\gog$ the positive roots of $\gog$ with respect to $\tilde \tot$, $\Phi_\gog^+$ a choice of positive roots, and $\tilde \tot_+ \subset \tilde \tot$ the positive Weyl chamber.  Since each $G$-orbit in $\gog$ intersects $\tilde \tot_+$ at a unique point, without loss of generality we may assume $\lambda \in \tilde \tot_+$.  Write
$$\Delta_\goh(x) = \prod_{\alpha \in \Phi^+_\goh} \langle \alpha, x \rangle, \qquad x \in \tot,$$
$$\Delta_\gog(x) = \prod_{\alpha \in \Phi^+_\gog} \langle \alpha, x \rangle, \qquad x \in \tilde \tot$$
for the products of the positive roots of $\goh$ and $\gog$ respectively, and
$$\rho_\goh = \frac{1}{2} \sum_{\alpha \in \Phi^+_\goh} \alpha, \qquad \rho_\gog = \frac{1}{2} \sum_{\alpha \in \Phi^+_\gog} \alpha$$
for the respective Weyl vectors.  Let $W_\goh$, $W_\gog$ be the Weyl groups generated by reflections in the root hyperplanes in $\tot$ and $\tilde \tot$ respectively.  Write $\epsilon(w)$ for the sign of $w \in W_\goh$ or $W_\gog$, that is, $\epsilon(w) = (-1)^{l(w)}$, where $l(w)$ is the number of reflections required to generate $w$.  Let $r$ be the rank of $\goh$, which is equal to $\dim \tot$ or equivalently to $\dim \goh - 2 |\Phi_\goh^+|$.

Each $\alpha \in \Phi_\goh^+$, regarded as a linear functional on $\tot$, is the restriction to $\tot$ of some root of $\gog$.  Accordingly, $\Delta_\goh$ divides the restriction $\Delta_\gog |_\tot$, so that $\Delta_\gog / \Delta_\goh$ is a polynomial on $\tot$, and we can define
\begin{equation} \label{eqn:delta-gh-def}
\Delta_{\gog / \goh}(x) = \frac{\Delta_\gog (x)}{\Delta_\goh(x)} = \prod_{\alpha \in \Phi^+_{\gog / \goh}} \langle \alpha, x \rangle, \qquad x \in \tot
\end{equation}
for some multiset of vectors $\Phi^+_{\gog / \goh} \subset \tot$.\footnote{Like the roots of $\goh$, the linear functionals $x \mapsto \langle \alpha, x \rangle$, $\alpha \in \Phi^+_{\gog / \goh}$ are restrictions to $\tot$ of roots of $\gog$.  They form a positive system of weights for the action of $H$ on $\gog / \goh$.  Elements of $\Phi^+_{\gog / \goh}$ may occur with multiplicity greater than 1, and we always count them with multiplicity.  As an example, for the two-body quantum marginal problem studied below in Section \ref{sec:formula-dist}, we find from (\ref{eqn:delta-rewrite}) that each positive root of $\mathfrak{su}(m)$ occurs with multiplicity $n-1$ and each positive root of $\mathfrak{su}(n)$ occurs with multiplicity $m-1$, in addition to the roots of the form $\big( \xi_j - \xi_p \pm (\zeta_k - \zeta_q) \big )$.}

For the sake of simplicity, we initially assume that $\Phi^+_{\gog/\goh}$ spans $\tot$, that $\Delta_\gog(\lambda) \ne 0$, and that the restriction of $\Delta_\gog$ to $\tot$ is not identically zero.  Under these assumptions, $\phi_* \beta_\lambda$ has a density $q_\lambda$ with respect to Lebesgue measure on $\goh$, and $\psi_* \phi_* \beta_\lambda$ has a density $p_\lambda$ with respect to Lebesgue measure on $\tot_+$.  Below in Remark \ref{rem:unfaithful}, we explain the role of these assumptions and discuss how the results of this subsection must be modified when they do not hold.

We have the following integral formula for the density $p_\lambda$.

\begin{thm} \label{thm:integral-formula-gen}
Let $\lambda \in \tilde \tot_+$ with $\Delta_\gog(\lambda) \ne 0$. Suppose that $\Phi^+_{\gog/\goh}$ spans $\tot$ and that $\Delta_\gog$ is not identically zero on $\tot$.  Then the density of $\psi_* \phi_* \beta_\lambda$ is equal to
\begin{equation} \label{eqn:integral-formula-gen}
p_\lambda(x) = C_{\gog / \goh} \frac{\Delta_\goh(x)}{\Delta_\gog(\lambda)} \int_\tot \frac{\Delta_\goh(\xi)}{\Delta_\gog(\xi)}  \sum_{w \in W_\gog} \epsilon(w) \, e^{i \langle \xi, w(\lambda) - x \rangle} \, d\xi
\end{equation}
for $x \in \tot_+$, where $d\xi$ is the Lebesgue measure on $\tot$ induced by the inner product, and
\begin{equation} \label{eqn:integral-constant-gen}
C_{\gog / \goh} = \frac{ (-i)^{|\Phi^+_\gog| - |\Phi^+_\goh|}}{(2\pi)^{r}} \frac{\Delta_\gog(\rho_\gog)}{\Delta_\goh(\rho_\goh)}.
\end{equation}
\end{thm}

\begin{remark} \label{rem:integral-convergence}
The integrand in (\ref{eqn:integral-formula-gen}) is an analytic function of $\xi$.  Although it may appear to diverge at points where $\Delta_\gog(\xi)$ vanishes, in fact at all such points the sum over $W_\gog$ vanishes as well, canceling the ostensible singularities.  However, for some choices of $G$ and $H$ the integral may not be absolutely convergent.  For example, Coquereaux and Zuber \cite{CZ1} observed that this occurs when $G = \SU(2) \times \SU(2)$ and $H$ is the diagonal subgroup.  In such cases, the formula (\ref{eqn:integral-formula-gen}) still holds, provided we interpret the integral as a Cauchy principal value.
\end{remark}

Before proving Theorem \ref{thm:integral-formula-gen}, we recall a well-known integral formula due to Harish-Chandra \cite{HC}:
\begin{equation} \label{eqn:hc-integral}
\int_G e^{\langle \mathrm{Ad}_g \lambda, \xi \rangle} dg = \frac{\Delta_\gog(\rho_\gog)}{\Delta_\gog(\lambda) \Delta_\gog(\xi)} \sum_{w \in W_\gog} \epsilon(w) \, e^{\langle w(\lambda), \xi \rangle}, \qquad \lambda, \xi \in \tilde \tot_\C,
\end{equation}
where $dg$ is the normalized Haar measure on $G$ and $\tilde \tot_\C$ is the complexification of $\tilde \tot$.  Similarly to the integrand in (\ref{eqn:integral-formula-gen}), although the right-hand side of (\ref{eqn:hc-integral}) is ostensibly singular as $\Delta_\gog(\lambda)$ or $\Delta_\gog(\xi)$ tends to zero, at all such values of $\lambda$ and $\xi$ the sum over $W_\gog$ vanishes to the same order as the denominator.  Accordingly these ostensible singularities are removable, and (\ref{eqn:hc-integral}) in fact defines a holomorphic function of $(\lambda, \xi) \in \tilde \tot_\C^2$.  When $G$ is the group $\U(N)$ of $N$-by-$N$ unitary matrices, (\ref{eqn:hc-integral}) gives the famous Harish-Chandra--Itzykson--Zuber (HCIZ) formula \cite{IZ}:
\begin{equation} \label{eqn:hciz} \int_{\U(N)} e^{\Tr (AUBU^\dagger)} dU = \left( \prod_{l=1}^{N-1}l! \right) \frac{\det \big[ e^{a_j b_k} \big]_{j,k = 1}^N}{\Delta_N(a) \Delta_N(b)}, \end{equation}
where $A$ and $B$ are two $N$-by-$N$ complex matrices with eigenvalues $a_1, \hdots, a_N$ and $b_1, \hdots, b_N$ respectively, $\Delta_N(a) = \prod_{j < k} (a_j - a_k)$ is the Vandermonde determinant, and $dU$ is the normalized Haar measure on $\U(N)$.  See \cite{McS-HC} for a detailed expository treatment of the formulae (\ref{eqn:hc-integral}) and (\ref{eqn:hciz}).

In what follows, we will write
\begin{equation} \label{eqn:integral-notation}
\mathcal{H}_G(\lambda, \xi) = \int_G e^{\langle \mathrm{Ad}_g \lambda, \xi \rangle} dg, \qquad \mathcal{H}_H(\lambda, \xi) = \int_H e^{\langle \mathrm{Ad}_h \lambda, \xi \rangle} dh, \qquad \lambda, \xi \in \tilde \tot_\C,
\end{equation}
where $dg$ and $dh$ are the respective normalized Haar measures on $G$ and $H$.  We use the probabilist's sign and normalization conventions for the Fourier transform of a measure $\mu$ on a Euclidean space $V$:
$$\mathscr{F}_V[\mu](\xi) = \int_V e^{i \langle x, \xi \rangle} \mu(dx).$$
When the space $V$ under consideration is clear, we omit the subscript.

\begin{proof}[Proof of Theorem \ref{thm:integral-formula-gen}]
The $G$-invariance of $\beta_\lambda$ implies that $\phi_* \beta_\lambda$ is $H$-invariant.  Accordingly, $p_\lambda$ is equal to the restriction of $q_\lambda$ to $\tot_+$ multiplied by the Jacobian of the map $\psi$.  By the Weyl integration formula, this Jacobian is equal to
\begin{equation} \label{eqn:psi-jacobian}
\kappa_\goh \Delta_\goh(x)^2, \qquad x \in \tot_+,
\end{equation}
where the normalizing constant $\kappa_\goh = (2\pi)^{|\Phi^+_\goh|} / \Delta_\goh(\rho_\goh)$ can be determined as described in \cite[\textsection1.1]{CMZ}.  Thus it suffices to compute $q_\lambda$, which we do by taking the Fourier transform of $\phi_*\beta_\lambda$ over $\goh$ and then the inverse Fourier transform: $q_\lambda(x) = \mathscr{F}^{-1}_\goh [\mathscr{F}_\goh [\phi_* \beta_\lambda]](x).$  First we compute
$$\mathscr{F}_\goh [\phi_* \beta_\lambda](\xi) = \int_\goh e^{i \langle x, \xi \rangle} \, \phi_* \beta_\lambda(dx) = \int_{\mathcal{O_\lambda}} e^{i \langle X, \xi \rangle} \, \beta_\lambda(dX) = \int_G e^{i \langle \mathrm{Ad}_g \lambda, \xi \rangle} dg.$$
Then taking the inverse Fourier transform, we find
\begin{align*}
q_\lambda(x) &= \frac{1}{(2\pi)^{\dim \goh}} \int_\goh e^{-i \langle x, \xi \rangle} \int_G e^{i \langle \mathrm{Ad}_g \lambda, \xi \rangle} dg \, d\xi \\
&= \frac{1}{(2\pi)^{\dim \goh}} \int_{\tot_+} \kappa_\goh \Delta_\goh(\xi)^2 \int_H e^{-i \langle x, \mathrm{Ad}_h \xi \rangle} \int_G e^{i \langle \mathrm{Ad}_g \lambda, \mathrm{Ad}_h \xi \rangle} dg \, dh \, d\xi \\
&= \frac{\kappa_\goh}{(2\pi)^{\dim \goh}} \int_{\tot_+} \Delta_\goh(\xi)^2 \int_H e^{-i \langle \mathrm{Ad}_h x, \xi \rangle} dh \int_G e^{i \langle \mathrm{Ad}_g \lambda, \xi \rangle} dg \, d\xi \\
&= \frac{\kappa_\goh}{(2\pi)^{\dim \goh} \, |W_\goh|} \int_{\tot} \Delta_\goh(\xi)^2 \, \mathcal{H}_G(\lambda, i \xi) \, \mathcal{H}_H(x, -i \xi) \, d\xi,
\end{align*}
where in the second line we have split the integral over $\goh$ into a radial integral over $\tot_+$ and an angular integral over $H$, and used the formula (\ref{eqn:psi-jacobian}) for the Jacobian of the map $\psi$, which is just the projection onto the radial coordinates.  In the third line, we have used the fact that the inner integrals over $G$ and $H$ are $H$-invariant in $x$, $\lambda$ and $\xi$, and in the fourth line we have used the fact that the entire integrand is $W_\goh$-invariant in $\xi$.

Next, restricting to $x \in \tot_+$ and multiplying by $\kappa_\goh \Delta_\goh(x)^2$, we find
\begin{equation} \label{eqn:p-orbital-integrals}
p_\lambda(x) = \frac{1}{(2\pi)^{r} \, |W_\goh|} \frac{\Delta_\goh(x)^2}{\Delta_\goh(\rho_\goh)^2} \int_{\tot} \Delta_\goh(\xi)^2 \, \mathcal{H}_H(x, -i \xi) \, \mathcal{H}_G(\lambda, i \xi) \, d\xi, \qquad x \in \tot_+,
\end{equation}
which is an interesting identity in its own right and essentially expresses $p_\lambda$ as an inverse Dunkl transform\footnote{See \cite{AnkerDunklNotes} for details on Dunkl transforms.} of $\mathcal{H}_G$. Then applying the Harish-Chandra formula (\ref{eqn:hc-integral}) to the integral over $H$, we obtain
\begin{multline*}
p_\lambda(x) =  \frac{1}{(2\pi)^{r} \, |W_\goh|} \frac{\Delta_\goh(x)^2}{\Delta_\goh(\rho_\goh)^2} \int_{\tot} \Delta_\goh(\xi)^2 \\
\times \bigg( \frac{\Delta_\goh(\rho_\goh)}{\Delta_\goh(x) \Delta_\goh(-i \xi)} \sum_{w \in W_\goh} \epsilon(w) \, e^{-i \langle w(x), \xi \rangle} \bigg) \mathcal{H}_G(\lambda, i \xi) \, d\xi, \qquad x \in \tot_+.
\end{multline*}
After cancelling redundant factors of $\Delta_\goh$ and using the fact that $\Delta_\goh$ is homogeneous of degree $|\Phi_\goh^+|$ to pull out factors of $i$, this becomes
\begin{align*}
p_\lambda(x) &=  \frac{i^{|\Phi_\goh^+|}}{(2\pi)^{r} \, |W_\goh|} \frac{\Delta_\goh(x)}{\Delta_\goh(\rho_\goh)} \int_{\tot} \Delta_\goh(\xi) \bigg( \sum_{w \in W_\goh} \epsilon(w) \, e^{-i \langle w(x), \xi \rangle} \bigg) \mathcal{H}_G(\lambda, i \xi) \, d\xi \\
&= \frac{i^{|\Phi_\goh^+|}}{(2\pi)^{r} \, |W_\goh|} \frac{\Delta_\goh(x)}{\Delta_\goh(\rho_\goh)} \sum_{w \in W_\goh} \epsilon(w) \int_{\tot} e^{-i \langle x, w(\xi) \rangle} \, \Delta_\goh(\xi) \, \mathcal{H}_G(\lambda, i \xi) \, d\xi.
\end{align*}
We next eliminate the sum over $W_\goh$ by changing integration variables from $\xi$ to $w(\xi)$ and using the following facts: the Lebesgue measure $d\xi$ is $W_\goh$-invariant, the orbital integral $\mathcal{H}_G(\lambda, i \xi)$ is $G$-invariant (and therefore $W_\goh$-invariant) in both arguments, the polynomial $\Delta_\goh$ is skew with respect to the action of $W$ (that is, $\Delta_\goh(w(\xi)) = \epsilon(w) \Delta_\goh(\xi)$), and $\epsilon(w^{-1}) = \epsilon(w)$.  Thus we obtain
\begin{align}
\label{eqn:p-almost-FT} p_\lambda(x) &= \frac{i^{|\Phi_\goh^+|}}{(2\pi)^{r} \, |W_\goh|} \frac{\Delta_\goh(x)}{\Delta_\goh(\rho_\goh)} \sum_{w \in W_\goh} \epsilon(w) \int_{\tot} e^{-i \langle x, \xi \rangle} \, \Delta_\goh(w^{-1}(\xi)) \, \mathcal{H}_G(\lambda, i w^{-1}(\xi)) \, d \xi \\
\nonumber &= \frac{i^{|\Phi_\goh^+|}}{(2\pi)^{r}} \frac{\Delta_\goh(x)}{\Delta_\goh(\rho_\goh)} \int_{\tot} e^{-i \langle x, \xi \rangle} \, \Delta_\goh(\xi) \, \mathcal{H}_G(\lambda, i \xi) \, d\xi.
\end{align}
Another application of (\ref{eqn:hc-integral}) to the integral over $G$ then yields (\ref{eqn:integral-formula-gen}).
\end{proof}

\begin{remark} \label{rem:unfaithful}
So far we have assumed that $\Phi^+_{\gog/\goh}$ spans $\tot$, that $\Delta_\gog(\lambda) \not = 0$, and that $\Delta_\gog$ is not identically equal to zero on $\tot$.  As we now explain, these assumptions guarantee that the density $p_\lambda$ actually exists and that the expression (\ref{eqn:integral-formula-gen}) is well defined.  All three assumptions are easily removed, at the cost of somewhat complicating the statement of Theorem \ref{thm:integral-formula-gen}.

First, suppose that $\psi_* \phi_* \beta_\lambda$ does indeed have a density on $\tot_+$ but that $\Delta_\gog(\lambda) = 0$ or $\Delta_\gog |_\tot \equiv 0$.  In this case the formula (\ref{eqn:p-almost-FT}) still holds, however we cannot apply (\ref{eqn:hc-integral}) directly to obtain (\ref{eqn:integral-formula-gen}) because the resulting expression is indeterminate.  Instead we must take a limit in (\ref{eqn:hc-integral}) and apply L'H\^opital's rule.  For example, we can write $\lambda_s = \lambda + s \rho_\gog$, $\xi_t = \xi + t \rho_\gog$; then $\Delta_\gog(\lambda_s) \ne 0$ and $\Delta_\gog(\xi_t) \ne 0$ for small $s, t > 0$.  Write $\partial_{\rho_\gog} \Delta_\gog(x) = \frac{d}{dt} \Delta_\gog(x + t \rho_\gog) |_{t = 0}.$ Then for some $j, k \ge 0$ we have
\begin{align}
\label{eqn:hc-hopital}
\mathcal{H}_G(\lambda, i\xi) &= \lim_{s,t \to 0} \frac{(-i)^{|\Phi^+_\gog|} \Delta_\gog(\rho_\gog)}{\Delta_\gog(\lambda_s) \Delta_\gog(\xi_t)} \sum_{w \in W_\gog} \epsilon(w) \, e^{i \langle w(\lambda_s), \xi_t \rangle} \\
\nonumber &= \frac{(-i)^{|\Phi^+_\gog|-j-k} \Delta_\gog(\rho_\gog)}{\partial_{\rho_\gog}^j \Delta_\gog(\lambda) \, \partial_{\rho_\gog}^k \Delta_\gog(\xi)} \sum_{w \in W_\gog} \epsilon(w) \, \langle w(\rho_\gog), \xi \rangle^j \, \langle w(\lambda), \rho_\gog \rangle^k \, e^{\langle w(\lambda), \xi \rangle},
\end{align}
where we have applied L'H\^opital's rule $j$ times in the limit as $s \to 0$ and $k$ times in the limit as $t \to 0$.  In place of (\ref{eqn:integral-formula-gen}), we then find
\begin{equation} \label{eqn:integral-formula-hopital}
p_\lambda(x) = C_{\gog / \goh} \frac{i^{j+k} \Delta_\goh(x)}{\partial_{\rho_\gog}^j \Delta_\gog(\lambda)} \int_\tot \frac{\Delta_\goh(\xi)}{\partial_{\rho_\gog}^k \Delta_\gog(\xi)}  \sum_{w \in W_\gog} \epsilon(w) \, \langle w(\rho_\gog), \xi \rangle^j \, \langle w(\lambda), \rho_\gog \rangle^k \, e^{i \langle \xi, w(\lambda) - x \rangle} \, d\xi.
\end{equation}

Second, we consider the existence of the density $p_\lambda$.  It follows from well-known properties of moment maps in symplectic geometry (see e.g.~\cite[ch.~1]{VG-comb}) that $\psi_* \phi_* \beta_\lambda$ has a density on $\tot_+$ if and only if $H$ acts locally freely at some point of $\mathcal{O}_\lambda$.  However, for our purposes a more practical criterion is that the density $p_\lambda$ exists as long as the integral in (\ref{eqn:p-almost-FT}) converges for all $x \in \tot_+$, possibly after applying a principal value prescription.  Under the other assumptions of Theorem \ref{thm:integral-formula-gen}, the requirement that $\Phi^+_{\gog / \goh}$ span $\tot$ ensures that the factor $\Delta_\goh(\xi) / \Delta_\gog(\xi) = \Delta_{\gog / \goh}(\xi)^{-1}$ in (\ref{eqn:integral-formula-gen}) decays quickly enough at infinity to guarantee this convergence.  (See Theorem \ref{thm:density-regularity} below for a more detailed statement relating the decay of $\Delta_{\gog / \goh}^{-1}$ and the regularity of $p_\lambda$.)

If the assumptions of Theorem \ref{thm:integral-formula-gen} do not hold, then in general $\psi_* \phi_* \beta_\lambda$ may be concentrated on a subset of positive codimension in $\tot$.  However, it will still have a density with respect to Lebesgue measure on the affine span of its support,\footnote{Concretely, $\psi_* \phi_* \beta_\lambda$ is concentrated on an affine subspace orthogonal to the Lie algebra of the subgroup of the maximal torus $\exp(\tot) \subset H$ that acts trivially on the subset $\phi^{-1}(\tot_+) \subset \mathcal{O}_\lambda$.} and in such cases a distributional analogue of Theorem \ref{thm:integral-formula-gen} still holds (see Corollary \ref{cor:pushfwd-FT} below).  Moreover, one can often use (\ref{eqn:integral-formula-gen}) or (\ref{eqn:integral-formula-hopital}) to derive a concrete integral formula for the density on this affine span.

We give two contrasting examples to illustrate how this may be done.  If we take $H = G$, then $\Delta_{\gog / \goh}^{-1} \equiv 1$.  The integral in (\ref{eqn:integral-formula-gen}) diverges when $x = \lambda$ and vanishes everywhere else on $\tot_+$.  Nevertheless, if we interpret (\ref{eqn:integral-formula-gen}) as a distributional identity for the measure $\psi_* \phi_* \beta_\lambda$, we correctly find
$$\psi_* \phi_* \beta_\lambda(dx) = \frac{1}{(2\pi)^{r}} \frac{\Delta_\gog(x)}{\Delta_\gog(\lambda)} \int_\tot \sum_{w \in W_\gog} \epsilon(w) \, e^{i \langle \xi, w(\lambda) - x \rangle} \, d\xi \bigg |_{x \in \tot_+} = \delta(x - \lambda),$$
a point mass concentrated at $\lambda$.  As a second example, suppose that there is a nontrivial subspace $V \subset \tot$ that is invariant under the action of $W_\gog$, as occurs in the quantum marginal problem studied in Section \ref{sec:formula-dist} below.  Then $V$ is orthogonal to the span of $\Phi^+_{\gog / \goh}$, and $\Delta_\gog$ and $\Delta_\goh$ are both invariant under translations by elements of $V$. Let $V^\perp$ be the orthogonal complement of $V$ in $\tot$. Suppose $\Phi^+_{\gog / \goh}$ spans $V^\perp$, and for $\xi \in \tot$, write $\xi = \xi_V + \xi_\perp$, with $\xi_V \in V$ and $\xi_\perp \in V^\perp$.  Then in place of (\ref{eqn:integral-formula-gen}), we have
\begin{align*}
\psi_* \phi_* \beta_\lambda(dx) &= C_{\gog / \goh} \frac{\Delta_\goh(x)}{\Delta_\gog(\lambda)} \int_{V^\perp} \int_V \frac{\Delta_\goh(\xi_V + \xi_\perp)}{\Delta_\gog(\xi_V + \xi_\perp)}  \sum_{w \in W_\gog} \epsilon(w) \, e^{i \langle \xi_V + \xi_\perp, w(\lambda) - x \rangle} \, d\xi_V \, d \xi_\perp \bigg |_{x \in \tot_+} \\
&= C_{\gog / \goh} \frac{\Delta_\goh(x)}{\Delta_\gog(\lambda)}  \int_V e^{i \langle \xi_V, \lambda - x \rangle} d\xi_V \int_{V^\perp}  \frac{\Delta_\goh(\xi_\perp)}{\Delta_\gog(\xi_\perp)}  \sum_{w \in W_\gog} \epsilon(w) \, e^{i \langle \xi_\perp, w(\lambda) - x \rangle} \, d \xi_\perp \bigg |_{x \in \tot_+},
\end{align*}
where the integral over $V$ in the second line is just a formal expression for the distribution $(2\pi)^{\dim V} \delta(x_V - \lambda_V)$. Thus we find that $\psi_* \phi_* \beta_\lambda$ has the density
$$ (2\pi)^{\dim V} C_{\gog / \goh} \frac{\Delta_\goh(x)}{\Delta_\gog(\lambda)} \int_{V^\perp}  \frac{\Delta_\goh(\xi_\perp)}{\Delta_\gog(\xi_\perp)}  \sum_{w \in W_\gog} \epsilon(w) \, e^{i \langle \xi_\perp, w(\lambda) - x \rangle} \, d \xi_\perp $$
with respect to Lebesgue measure on $(\lambda_V + V^\perp) \cap \tot_+$.

In this way, by applying L'H\^opital's rule and passing to a lower-dimensional affine subspace if necessary, we can relax the assumptions of Theorem \ref{thm:integral-formula-gen}.
\end{remark}

If a function $f : \tot \to \C$ is locally integrable and has at most polynomial growth at infinity, then its inverse Fourier transform
$$\mathscr{F}_\tot^{-1}[f](x) = \frac{1}{(2\pi)^{r}} \int_\tot f(\xi) \, e^{-i \langle x, \xi \rangle} d\xi$$
exists as a tempered distribution.  To end this subsection, we note that (\ref{eqn:p-almost-FT}) can be reformulated as the following compact expression for $\psi_* \phi_* \beta_\lambda$, which holds in a distributional sense without any assumptions on $\Phi^+_{\gog / \goh}$, $\Delta_\gog$ or $\lambda$.  

\begin{cor} \label{cor:pushfwd-FT}
\begin{equation} \label{eqn:pushfwd-FT}
\psi_* \phi_* \beta_\lambda(dx) = \frac{\Delta_\goh(x)}{\Delta_\goh(\rho_\goh)} \mathscr{F}_\tot^{-1}\big[ \Delta_\goh( i \bullet ) \, \mathcal{H}_G(\lambda, i \bullet ) \big] \Big |_{\tot_+}.
\end{equation}
\end{cor}

Applying a standard Fourier transform identity to (\ref{eqn:pushfwd-FT}), we find that the argument above gives an alternate proof of the following ``derivative principle'':
\begin{equation} \label{eqn:deriv-princ}
\psi_* \phi_* \beta_\lambda(dx) = \frac{\Delta_\goh(x)}{\Delta_\goh(\rho_\goh)}  \Delta_\goh( -\partial ) \mathscr{F}_\tot^{-1}\big[ \mathcal{H}_G(\lambda, i \bullet ) \big] \Big |_{\tot_+},
\end{equation}
where
$$\Delta_\goh( -\partial ) = \prod_{\alpha \in \Phi^+_\goh}(-\partial_\alpha)$$
and $\partial_\alpha$ is the derivative in the direction of $\alpha$, i.e. $\partial_\alpha f(x) = \frac{d}{dt} f(x + t \alpha)|_{t=0}$. The distributional identity (\ref{eqn:deriv-princ}) is essentially equivalent to \cite[Lemma 6.1(c)]{Heck} and was also proved in \cite[Corollary 3.3]{CDKW}.  Beyond the context of Lie algebras, analogous derivative principles for certain other symmetric spaces have been studied in \cite{KZ-deriv}.

\subsection{Distinguishable particles} \label{sec:formula-dist}

We now apply the above ideas to the quantum marginal problem for distinguishable particles, starting with the two-body case described in the introduction.  Although the main result of this section, Corollary \ref{cor:density-formula} below, can be deduced from Theorem \ref{thm:integral-formula-gen}, here we derive the formula directly without using any Lie theory, in order to provide a more concrete illustration of the abstract considerations in Section \ref{sec:projections-general}.  In Remark \ref{rem:distinguishable-lietheory} we explain how the setup below relates to the setup of Theorem \ref{thm:integral-formula-gen}.

Let $m, n \ge 2$.\footnote{Note that when $m=1$ or $n=1$, the two-body quantum marginal problem is trivial: when $m=1$, for example, we have $\pi_1(M) = \Tr M$, $\pi_2(M) = M$.}  Given $\lambda = (\lambda_{11}, \hdots, \lambda_{mn}) \in \R^{mn}$, with the double indices on the coordinates $\lambda_{jk}$ ordered lexicographically, let $\Lambda = \mathrm{diag}(\lambda)$ and write $$\mathcal{O}_\lambda = \{ U \Lambda U^\dagger \ | \ U \in \U(mn) \}$$ for the orbit of $\Lambda$ under conjugation by elements of $\U(mn)$.  The orbit $\mathcal{O}_\lambda$ is a compact submanifold of the space $\mathrm{Her}(mn)$ of $mn$-by-$mn$ Hermitian matrices, and it carries a unique $\U(mn)$-invariant probability measure, the orbital measure $\beta_\lambda$.  Let $\phi: \mathrm{Her}(mn) \to \mathrm{Her}(m) \oplus \mathrm{Her}(n)$ be the map $M \mapsto (\pi_1(M), \pi_2(M))$, with $\pi_1$ and $\pi_2$ as defined in (\ref{eqn:marginal-def}), and let $\psi: \mathrm{Her}(m) \oplus \mathrm{Her}(n) \to \R^{m+n}$ be the map that sends a pair of matrices $(M,N)$ to the pair $(\mu, \nu)$, where $\mu$ and $\nu$ are respectively the vectors of eigenvalues of $M$ and $N$, sorted in non-increasing order.  We endow the spaces $\mathrm{Her}(m)$, $\mathrm{Her}(n)$ and $\mathrm{Her}(mn)$ with the Frobenius inner product $\langle A, B \rangle = \Tr(AB)$.

The pushforward measure $\psi_* \phi_* \beta_\lambda$ is the joint distribution of the eigenvalues of $\pi_1(M)$ and $\pi_2(M)$, where $M$ is uniformly (i.e.~invariantly) distributed on $\mathcal{O}_\lambda$.  Since $\Tr \pi_1(M) = \Tr \pi_2(M) = \Tr \Lambda,$ the support of $\psi_* \phi_* \beta_\lambda$ is contained in the codimension-2 affine subspace
$$ \Big \{ (\mu, \nu) \in \R^{m + n} \ \Big | \  \sum_{j=1}^m \mu_j = \sum_{k=1}^n \nu_k = \sum_{jk=11}^{mn} \lambda_{jk} \Big \}.$$
In fact, we can assume without loss of generality that $\Tr \Lambda = 0$, since, as explained below in Remark \ref{rem:trace-zero}, changing the trace of $\Lambda$ merely translates the support of $\psi_* \phi_* \beta_\lambda$. When all $\lambda_{jk}$ are distinct, $\psi_* \phi_* \beta_\lambda$ then has a density $p_\lambda^{\mathrm{dst}}(\mu, \nu)$ with respect to Lebesgue measure on the $(m + n - 2)$-dimensional subspace 
$$ \Sigma_0 = \Big \{ (\mu, \nu) \in \R^{m + n} \ \Big | \  \sum_{j=1}^m \mu_j = \sum_{k=1}^n \nu_k = 0 \Big \}.$$

Here we give an integral formula for $p^{\mathrm{dst}}_\lambda(\mu, \nu)$.  (The superscript ``dst'' stands for ``distinguishable.'') For $\lambda \in \R^{mn}$, write $\Delta_{mn}(\lambda) = \prod_{ij < kl}(\lambda_{ij} - \lambda_{kl}),$ where $ij < kl$ means that the double index $ij$ precedes the double index $kl$ in lexicographic order.  Define
$$\Sigma_0^\downarrow = \big \{ (\mu, \nu) \in \Sigma_0 \ \big | \ \mu_1 \ge \hdots \ge \mu_m, \ \nu_1 \ge \hdots \ge \nu_n \big \},$$
and write $\mathbbm{1}_{\Sigma_0^\downarrow}(\mu, \nu)$ for the indicator function of $\Sigma_0^\downarrow$.

\begin{cor} \label{cor:density-formula}
Take $\lambda \in \R^{mn}$ with $\lambda_{11} > \hdots > \lambda_{mn}$ and $\sum_{jk=11}^{mn} \lambda_{jk} = 0$. Then
\begin{multline} \label{eqn:density-formula}
p^{\mathrm{dst}}_\lambda(\mu, \nu) = C^\mathrm{dst}_{m,n} \mathbbm{1}_{\Sigma_0^\downarrow}(\mu, \nu) \frac{\Delta_m(\mu) \Delta_n(\nu)}{\Delta_{mn}(\lambda)} \int_{\Sigma_0} \\
\frac{ \Delta_m(\xi) \Delta_n(\zeta)}{ \Delta_{mn}([\xi_p + \zeta_q]_{pq=11}^{mn})} \, e^{-i (\sum_{j=1}^m \mu_j \xi_j + \sum_{k=1}^n \nu_k \zeta_k )} \, \det \big[ e^{i \lambda_{jk} (\xi_p + \zeta_q)} \big]_{jk,pq = 11}^{mn} \\
dL(\xi_1, \hdots, \xi_m, \zeta_1, \hdots, \zeta_n),
\end{multline}
where $dL$ is Lebesgue measure on $\Sigma_0$ and
\begin{equation} \label{eqn:kappa-def}
C^\mathrm{dst}_{m,n} = \frac{(-i)^{\frac{1}{2}(mn(mn-1) - m(m-1) - n(n-1))}}{(2\pi)^{m+n-2}} \frac{\prod_{l=1}^{mn-1}l!}{\left( \prod_{l=1}^{m-1} l! \right) \left( \prod_{l=1}^{n-1} l! \right)}.
\end{equation}
\end{cor}

\begin{remark} \label{rem:distinguishable-lietheory}
Before giving the direct proof of Corollary \ref{cor:density-formula}, we briefly describe the statement in the terms of Theorem \ref{thm:integral-formula-gen}.  The Lie algebra of the group $G = \U(mn)$ is $\mathfrak{u}(mn)$, the space of $mn$-by-$mn$ {\it anti}-Hermitian matrices.  If $B \in \mathfrak{u}(mn)$ then $i B \in \mathrm{Her}(mn)$, giving a natural identification between $\mathfrak{u}(mn)$ and $\mathrm{Her}(mn)$.  Under this identification, the conjugation action of $\U(mn)$ on $\mathrm{Her}(mn)$ coincides with the adjoint representation of $\U(mn)$ on $\mathfrak{u}(mn)$, and $\mathcal{O}_\lambda$ is an adjoint orbit.  (Equivalently, we could identify $\mathrm{Her}(mn)$ with the dual space $\mathfrak{u}(mn)^*$ via the pairing $(A, B) = \Tr(iAB)$ for $A \in \mathrm{Her}(mn)$, $B \in \mathfrak{u}(mn)$.)

Let $\goh$ be the Lie algebra of the subgroup $H = \U(m) \otimes \U(n) \subset \U(mn)$.  Let $\mathbbm{1}_m$ and $\mathbbm{1}_n$ denote the $m$-by-$m$ and $n$-by-$n$ identity matrices.  Given $A \in \mathrm{Her}(m)$ and $B \in \mathrm{Her}(n)$, write $A \otimes B \in \mathrm{Her}(mn)$ for their Kronecker product. Then
$$\goh \cong \big \{ \xi \otimes \mathbbm{1}_n + \mathbbm{1}_m \otimes \zeta \ | \ \xi \in \mathrm{Her}(m), \, \zeta \in \mathrm{Her}(n) \big \} \subset \mathrm{Her}(mn).$$
The orthogonal projection from $\mathrm{Her}(mn)$ to $\goh$ is given by
\begin{equation} \label{eqn:mm-formula}
M \mapsto \frac{1}{n} \pi_1\bigg( M - \frac{\Tr M}{mn} \mathbbm{1}_{mn} \bigg) \otimes \mathbbm{1}_n + \frac{1}{m} \mathbbm{1}_m \otimes \pi_2\bigg( M - \frac{\Tr M}{mn} \mathbbm{1}_{mn} \bigg) + \frac{\Tr M}{mn} \mathbbm{1}_{mn}.
\end{equation}
We take $\tilde \tot \subset \mathrm{Her}(mn)$ to be the space of $mn$-by-$mn$ real diagonal matrices, and $\tot \cong \R^{m+n-1}$ to be the space of matrices $\xi \otimes \mathbbm{1}_n + \mathbbm{1}_m \otimes \zeta \in \goh$ with $\xi, \zeta$ both diagonal.  The positive chamber $\tot_+$ consists of the matrices $\xi \otimes \mathbbm{1}_n + \mathbbm{1}_m \otimes \zeta \in \tot$ such that the entries of $\xi$ and $\zeta$ are non-increasing down their respective diagonals, and we have an obvious embedding $\Sigma_0 \hookrightarrow \tot$ such that $\Sigma_0^\downarrow = \Sigma_0 \cap \tot_+$.  The Weyl group of $\U(mn)$ is the symmetric group $S_{mn}$ and acts on $\tilde \tot$ by permuting the diagonal entries.  The subspace of $\tot$ spanned by $\mathbbm{1}_{mn}$ is invariant under this action; as discussed in Remark \ref{rem:unfaithful}, $\psi_* \phi_* \beta_\lambda$ is therefore supported on an affine subspace orthogonal to $\mathbbm{1}_{mn}$.  This corresponds to the linear constraint $\Tr \pi_1(M) = \Tr \pi_2(M) = \Tr \Lambda$.  We thus restrict our attention to the subspaces $\Sigma_0$ and
$\psi^{-1}(\Sigma_0) = \{ (X, Y) \ | \ \Tr X = \Tr Y = 0 \} \subset \mathrm{Her}(m) \oplus \mathrm{Her}(n).$
Observing that
$$\pi_1(\xi \otimes \mathbbm{1}_n + \mathbbm{1}_m \otimes \zeta) = n \xi + (\Tr \zeta) \mathbbm{1}_m,$$
$$\pi_2(\xi \otimes \mathbbm{1}_n + \mathbbm{1}_m \otimes \zeta) = m \zeta + (\Tr \xi) \mathbbm{1}_n,$$
we find that there is a linear bijection
\begin{equation} \label{eqn:correspondence-k}
\xi \otimes \mathbbm{1}_n + \mathbbm{1}_m \otimes \zeta \ \longleftrightarrow \ (n \xi, m \zeta)
\end{equation}
between $\psi^{-1}(\Sigma_0)$ and the subspace
$$\big \{ \xi \otimes \mathbbm{1}_n + \mathbbm{1}_m \otimes \zeta \ | \ \Tr \xi = \Tr \zeta = 0 \big \} \subset \goh.$$
If we use this bijection to identify these two subspaces, and if we assume $\sum \lambda_{jk} = 0$, then the map $\phi : M \mapsto (\pi_1(M), \pi_2(M))$ coincides with the orthogonal projection (\ref{eqn:mm-formula}) on $\mathcal{O}_\lambda$, so that the measure $\psi_* \phi_* \beta_\lambda$ is indeed the same as in Section \ref{sec:projections-general} in the case $G = \U(mn)$, $H = \U(m) \otimes \U(n)$.

Finally, note that the orbit of a Hermitian matrix under conjugation by unitary matrices is the same as its orbit under conjugation by unitary matrices with determinant 1.  Therefore, although the quantum marginal problem is typically framed in terms of unitary orbits of Hermitian matrices, we could equivalently take $G = \SU(mn)$ and $H = \SU(m) \otimes \SU(n)$, and consider their actions on the space of traceless Hermitian matrices.  In later sections we will take this approach, as it lets us avoid explicitly restricting to $\Sigma_0$.
\end{remark}

\begin{proof}[Proof of Corollary \ref{cor:density-formula}]
We first note that the measure $\phi_* \beta_\lambda$ is supported on the preimage
\begin{equation} \label{eqn:sigma0-preimage}
\psi^{-1}(\Sigma_0) = \big \{ (X, Y) \ | \ \Tr X = \Tr Y = 0 \big \} \subset \mathrm{Her}(m) \oplus \mathrm{Her}(n).
\end{equation}
If we identify $\psi^{-1}(\Sigma_0)$ with its dual using the Frobenius inner product, then the Fourier transform of $\phi_* \beta_\lambda$ on this subspace is also a function on $\psi^{-1}(\Sigma_0)$ and is given by:
\begin{align*}
\mathscr{F}[ \phi_* \beta_\lambda ](\xi, \zeta) &= \int_{\psi^{-1}(\Sigma_0)} \exp \big[i \Tr(\xi A) + i \Tr(\zeta B) \big] \, \phi_* \beta_\lambda(dA\, dB) \\
& = \int_{\mathcal{O}_\lambda} \exp\big[ i \Tr(\xi \pi_1(X)) + i \Tr(\zeta \pi_2(X))\big] \, \beta_\lambda(dX) \\
&= \int_{\mathcal{O}_\lambda} \exp\big[i \Tr( (\xi \otimes \mathbbm{1}_n + \mathbbm{1}_m \otimes \zeta) X) \big] \, \beta_\lambda(dX) \\
&= \int_{\U(mn)} \exp\big[i \Tr( (\xi \otimes \mathbbm{1}_n + \mathbbm{1}_m \otimes \zeta) U \Lambda U^\dagger) \big] \, dU,
\end{align*}
where $dU$ is the normalized Haar measure on $\U(mn)$.  Write $\xi_1 \ge \hdots \ge \xi_m$ and $\zeta_1 \ge \hdots \ge \zeta_n$ for the eigenvalues of $\xi$ and $\zeta$ respectively.  Then the eigenvalues of $\xi \otimes \mathbbm{1}_n + \mathbbm{1}_m \otimes \zeta$ are $\xi_p + \zeta_q$, $1 \le p \le m$, $1 \le q \le n$, and the HCIZ formula (\ref{eqn:hciz}) gives:
\begin{equation} \label{eqn:FT-initial}
\mathscr{F}[ \phi_* \beta_\lambda ](\xi, \zeta) = \left( \prod_{l=1}^{mn-1}l! \right) \frac{\det \big[ e^{i \lambda_{jk} (\xi_p + \zeta_q)} \big]_{jk,pq = 11}^{mn}}{\Delta_{mn}(i \lambda) \Delta_{mn}([\xi_p + \zeta_q]_{pq=11}^{mn})}.
\end{equation}
Taking the inverse Fourier transform of the above expression and using the facts that $\dim \psi^{-1}(\Sigma_0) = m^2 + n^2 - 2$ and that $\Delta_{mn}$ is homogeneous of degree $mn(mn-1)/2$, we find that the density of $\phi_* \beta_\lambda$ is given by
\begin{multline*}
\mathscr{F}^{-1}[\mathscr{F}[ \phi_* \beta_\lambda ]](X,Y) = \\ \frac{(-i)^{mn(mn-1)/2} \prod_{l=1}^{mn-1}l!}{ (2\pi)^{m^2 + n^2 - 2}}  \int_{\psi^{-1}(\Sigma_0)} \frac{\det \big[ e^{i \lambda_{jk} (\xi_p + \zeta_q)} \big]_{jk,pq = 11}^{mn} }{\Delta_{mn}(\lambda) \Delta_{mn}([\xi_p + \zeta_q]_{pq=11}^{mn})} \, e^{-i \Tr(X \xi) - i \Tr(Y \zeta)} \, d\xi \, d\zeta,
\end{multline*}
where $d\xi \, d\zeta$ indicates the Lebesgue measure induced by the Frobenius inner product on $\psi^{-1}(\Sigma_0)$.

Now we change to generalized polar coordinates.  That is, we split the above integral over $\psi^{-1}(\Sigma_0)$ into multiple iterated integrals: an outer radial integral over the eigenvalues $(\xi_1, \hdots, \xi_m, \zeta_1, \hdots, \zeta_n) \in \Sigma_0$, and two inner angular integrals over the $\U(m)$-orbit of $\xi$ and the $\U(n)$-orbit of $\zeta$.  Writing $\mathcal{O}_\xi$, $\mathcal{O}_\zeta$ respectively for these two orbits, we have
\begin{multline} \label{eqn:FT-intermediate1}
\mathscr{F}^{-1}[\mathscr{F}[ \phi_* \beta_\lambda ]](X,Y) = \frac{(-i)^{mn(mn-1)/2} \prod_{l=1}^{mn-1}l!}{ (2\pi)^{m^2 + n^2 - 2} m! n!}  \int_{\Sigma_0} \\ \frac{\det \big[ e^{i \lambda_{jk} (\xi_p + \zeta_q)} \big]_{jk,pq = 11}^{mn} }{\Delta_{mn}(\lambda) \Delta_{mn}([\xi_p + \zeta_q]_{pq=11}^{mn})} \int_{\mathcal{O_\xi}} e^{-i \Tr(X A)} dA \int_{\mathcal{O}_\zeta} e^{- i \Tr(Y B)} dB \\ dL(\xi_1, \hdots, \xi_m, \zeta_1, \hdots, \zeta_n),
\end{multline}
where $dA$ and $dB$ are the volume forms induced on $\mathcal{O}_\xi$ and $\mathcal{O}_\zeta$ as submanifolds of $\mathrm{Her}(m)$ and $\mathrm{Her}(n)$ respectively.  The factor of $m! n!$ in the denominator of the leading constant in (\ref{eqn:FT-intermediate1}) is included to account for the fact that we are integrating over the entire subspace $\Sigma_0$ and not merely over $\Sigma_0^\downarrow$.

Away from a set of measure zero in $\Sigma_0$, the eigenvalues $\xi_j$ and $\zeta_k$ all take distinct values.  Then the volumes of the orbits $\mathcal{O}_\xi$ and $\mathcal{O}_\nu$ are given by:\footnote{These volumes are special cases of (\ref{eqn:psi-jacobian}) and are just the standard Jacobian factors used to write an invariant probability density on $\mathrm{Her}(m)$ or $\mathrm{Her}(n)$ in terms of the eigenvalues.  See e.g.~\cite[Appendix A]{Z}.}
\begin{equation} \label{eqn:orbit-volumes}
\mathrm{Vol}(\mathcal{O}_\xi) = \frac{(2\pi)^{m(m-1)/2}}{\prod_{l=1}^{m-1} l!} \Delta_m(\xi)^2, \qquad \mathrm{Vol}(\mathcal{O}_\zeta) = \frac{(2\pi)^{n(n-1)/2}}{\prod_{l=1}^{n-1} l!} \Delta_n(\zeta)^2,
\end{equation}
where $\Delta_m(\xi) = \prod_{1 \le i < j \le m} (\xi_i - \xi_j)$ and $\Delta_n(\zeta) = \prod_{1 \le i < j \le n} (\zeta_i - \zeta_j)$. Using (\ref{eqn:orbit-volumes}), we find
$$\int_{\mathcal{O_\xi}} e^{-i \Tr(X A)} dA = \frac{(2\pi)^{m(m-1)/2}}{\prod_{l=1}^{m-1} l!} \Delta_m(\xi)^2 \int_{\U(m)} e^{-i \Tr(X U\xi U^\dagger)} dU,$$
where $dU$ is the normalized Haar measure on $\U(m)$, and a similar expression for the integral over $\mathcal{O}_\zeta$ in (\ref{eqn:FT-intermediate1}).  Thus we can rewrite (\ref{eqn:FT-intermediate1}) as
\begin{multline*}
\mathscr{F}^{-1}[\mathscr{F}[ \phi_* \beta_\lambda ]](X,Y) = \frac{(-i)^{mn(mn-1)/2} \prod_{l=1}^{mn-1}l!}{ (2\pi)^{\frac{1}{2}(m(m+1)+n(n+1) - 4)} m!n! \left( \prod_{l=1}^{m-1} l! \right) \left( \prod_{l=1}^{n-1} l! \right)}  \int_{\Sigma_0} \\ \frac{\det \big[ e^{i \lambda_{jk} (\xi_p + \zeta_q)} \big]_{jk,pq = 11}^{mn} \Delta_m(\xi)^2 \Delta_n(\zeta)^2}{\Delta_{mn}(\lambda) \Delta_{mn}([\xi_p + \zeta_q]_{pq=11}^{mn})} \int_{\U(m)} e^{-i \Tr(X U\xi U^\dagger)} dU \int_{\U(n)} e^{- i \Tr(Y V \zeta V^\dagger)} dV \\ dL(\xi_1, \hdots, \xi_m, \zeta_1, \hdots, \zeta_n).
\end{multline*}
Writing $\mu_1 \ge \hdots \ge \mu_m$ for the eigenvalues of $X$ and $\nu_1 \ge \hdots \ge \nu_n$ for the eigenvalues of $Y$, two more applications of the HCIZ formula then yield
\begin{multline*}
\mathscr{F}^{-1}[\mathscr{F}[ \phi_* \beta_\lambda ]](X,Y) = \frac{(-i)^{\frac{1}{2}(mn(mn-1) - m(m-1) - n(n-1))} \prod_{l=1}^{mn-1}l!}{(2\pi)^{\frac{1}{2}(m(m+1)+n(n+1) - 4)} \Delta_{mn}(\lambda) \Delta_m(\mu) \Delta_n(\nu) m! n!}  \int_{\Sigma_0} \\ \frac{\Delta_m(\xi) \Delta_n(\zeta)}{ \Delta_{mn}([\xi_p + \zeta_q]_{pq=11}^{mn}) } \det \big[ e^{i \lambda_{jk} (\xi_p + \zeta_q)} \big]_{jk,pq = 11}^{mn} \det \big[ e^{-i \mu_j \xi_k} \big]_{j,k = 1}^m \det \big[ e^{-i \nu_p \zeta_q} \big]_{p,q = 1}^n \\ dL(\xi_1, \hdots, \xi_m, \zeta_1, \hdots, \zeta_n).
\end{multline*}
The integrand above contains three determinants.  If we expand the latter two determinants as sums of permutations, we find
\begin{multline*}
\mathscr{F}^{-1}[\mathscr{F}[ \phi_* \beta_\lambda ]](X,Y) = \frac{(-i)^{\frac{1}{2}(mn(mn-1) - m(m-1) - n(n-1))} \prod_{l=1}^{mn-1}l!}{(2\pi)^{\frac{1}{2}(m(m+1)+n(n+1) - 4)} \Delta_{mn}(\lambda) \Delta_m(\mu) \Delta_n(\nu) m! n!}  \int_{\Sigma_0} \\ \bigg( \frac{\Delta_m(\xi) \Delta_n(\zeta)}{ \Delta_{mn}([\xi_p + \zeta_q]_{pq=11}^{mn}) } \det \big[ e^{i \lambda_{jk} (\xi_p + \zeta_q)} \big]_{jk,pq = 11}^{mn} \bigg) \sum_{\substack{\sigma \in S_m \\ \tau \in S_n}} e^{-i(\sum_{j=1}^m \mu_j \xi_{\sigma(j)} + \sum_{k=1}^n \nu_k \zeta_{\tau(k)})} \\ dL(\xi_1, \hdots, \xi_m, \zeta_1, \hdots, \zeta_n).
\end{multline*}
Pulling the sum out of the integral and using the fact that the quantity in parentheses in the second line is skew with respect to the actions of $S_m$ on $\xi$ and $S_n$ on $\zeta$, we finally arrive at
\begin{multline} \label{eqn:phipushfwd-density}
\mathscr{F}^{-1}[\mathscr{F}[ \phi_* \beta_\lambda ]](X,Y) = \frac{(-i)^{\frac{1}{2}(mn(mn-1) - m(m-1) - n(n-1))} \prod_{l=1}^{mn-1}l!}{(2\pi)^{\frac{1}{2}(m(m+1)+n(n+1) - 4)} \Delta_{mn}(\lambda) \Delta_m(\mu) \Delta_n(\nu)}  \int_{\Sigma_0} \\ \frac{\Delta_m(\xi) \Delta_n(\zeta)}{ \Delta_{mn}([\xi_p + \zeta_q]_{pq=11}^{mn}) } \det \big[ e^{i \lambda_{jk} (\xi_p + \zeta_q)} \big]_{jk,pq = 11}^{mn} e^{-i(\sum_{j=1}^m \mu_j \xi_j + \sum_{k=1}^n \nu_k \zeta_k)} \\ dL(\xi_1, \hdots, \xi_m, \zeta_1, \hdots, \zeta_n).
\end{multline}

The above expression gives the density of $\phi_* \beta_\lambda$ and is therefore a function on $\psi^{-1}(\Sigma_0) \subset \mathrm{Her}(m) \oplus \mathrm{Her}(n)$.  However, we have now written it in terms of the eigenvalues $(\mu,\nu) \in \Sigma_0^\downarrow$.  To obtain the density of $\psi_*\phi_* \beta_\lambda$, it remains only to multiply by the Jacobian of the map $\psi$.  This Jacobian factor is equal to $\mathbbm{1}_{\Sigma_0^\downarrow}(\mu,\nu) \mathrm{Vol}(\mathcal{O}_\mu) \mathrm{Vol}(\mathcal{O}_\nu)$.  Again using (\ref{eqn:orbit-volumes}), we find
\begin{equation*}
\mathrm{Vol}(\mathcal{O}_\mu) \mathrm{Vol}(\mathcal{O}_\nu) = \frac{(2\pi)^{(m(m-1)+n(n-1))/2}}{\left( \prod_{l=1}^{m-1} l! \right) \left( \prod_{l=1}^{n-1} l! \right) } \Delta_m(\mu)^2 \Delta_n(\nu)^2,
\end{equation*}
which gives the desired formula (\ref{eqn:density-formula}).
\end{proof}

\begin{remark} \label{rem:trace-zero}
The assumption that $\sum \lambda_{jk} = 0$ is justified by the following observation.  Take $t \in \R$ and let $\lambda_{jk}' = \lambda_{jk} + t/(mn)$ for $jk = 11, \hdots, mn$.  Then $\psi_* \phi_* \beta_{\lambda'}$ is supported on the affine subspace
$$ \Sigma_t = \Big \{ (\mu', \nu') \in \R^{m + n} \ \Big | \  \sum_{j=1}^m \mu'_j = \sum_{k=1}^n \nu'_k = t \Big \}, $$
but we can identify $\Sigma_t \cong \Sigma_0$ via the change of variables $\mu_j = \mu'_j - nt/m$, $\nu_k = \nu'_k - mt/n$, for $j = 1, \hdots, m$ and $k=1, \hdots, n$. By translating $\Sigma_t$ in this way, we can still regard the density of $\psi_* \phi_* \beta_{\lambda'}$ as a function of $(\mu, \nu) \in \Sigma_0$ rather than $(\mu', \nu') \in \Sigma_t$.  After making this change of variables, (\ref{eqn:FT-initial}) gives
$$ \mathscr{F}[ \phi_* \beta_{\lambda'} ](\xi, \zeta) = e^{i t \sum_{p,q}(\xi_p + \zeta_q)} \mathscr{F}[ \phi_* \beta_\lambda ](\xi, \zeta) = \mathscr{F}[ \phi_* \beta_\lambda ], $$
since $\sum_{p,q}(\xi_p + \zeta_q) = 0$.  Therefore $\psi_* \phi_* \beta_{\lambda'}$ and $\psi_* \phi_* \beta_{\lambda}$ differ only by the translation $(\mu, \nu) \mapsto (\mu', \nu')$, and accordingly we lose no generality by considering only $\lambda$ with $\sum \lambda_{jk} = 0$.
\end{remark}

\begin{remark} \label{rem:purestates}
Although Corollary \ref{cor:density-formula} holds for generic $\lambda$, an important situation in which the formula (\ref{eqn:density-formula}) does {\it not} apply as written, even after translating to the traceless subspace, is the case of so-called pure states, corresponding to the choice $\lambda = (1,0,\hdots,0).$  For pure states, the orbit $\mathcal{O}_\lambda$ is isomorphic to the complex projective space $\mathbb{P}^{mn}$, and --- assuming without loss of generality that $m \le n$ --- one can show by an elementary calculation that the measure $\psi_* \phi_* \beta_\lambda$ is supported on the affine subspace consisting of $(\mu, \nu)$ satisfying $\mu_j = \nu_j$ for $1 \le j \le m$ and $\nu_j = 0$ for $m < j \le n$.  In this case, combining Corollary \ref{cor:pushfwd-FT} with the explicit form of the HCIZ integral for projective spaces shown in \cite[Proposition 7.4]{LeakeVishnoi} yields the following distributional identity:
\begin{equation} \label{eqn:purestates-dist}
\psi_* \phi_* \beta_\lambda(d\mu \, d\nu) = C_{m,n}^\mathrm{pure} \, \Delta_m(\mu) \, \Delta_n(\nu) \, \mathscr{F}_{\Sigma_0}^{-1}\bigg[ \frac{\Delta_m(\xi) \Delta_n(\zeta) \det M(\xi, \zeta)}{\Delta_{mn}([\xi_p + \zeta_q]_{pq=11}^{mn})} \bigg] \Bigg |_{\Sigma_0^\downarrow},
\end{equation}
where $(\xi, \zeta)$ are the integration variables under the Fourier transform,
\begin{equation} \label{eqn:purestates-const}
C_{m,n}^\mathrm{pure} = \frac{i^{\frac{1}{2}(m(m-1) + n(n-1))}(mn-1)!}{\left( \prod_{l=1}^{m-1} l! \right) \left( \prod_{l=1}^{n-1} l! \right)},
\end{equation}
and $M(\xi, \zeta)$ is the $mn$-by-$mn$ matrix with entries
\begin{equation} \label{eqn:purestates-Mdef}
M(\xi, \zeta)_{ij,kl} = \begin{cases}
(\xi_i + \zeta_j)^{(m-1)k + l - 1}, & kl \ne mn, \\
e^{\xi_i + \zeta_j}, & kl = mn.
\end{cases}
\end{equation}
In the case of pure states, we leave open the question of upgrading (\ref{eqn:purestates-dist}) to an integral formula for the pointwise values of the density of $\psi_* \phi_* \beta_\lambda$ on the affine span of its support.
\end{remark}

The preceding discussion easily generalizes to the quantum marginal problem for arbitrarily many distinguishable particles, which we now describe.  Let $M \in \mathrm{Her}(N)$ for $N = \prod_{j=1}^m n_j$.  We can regard $M$ as an operator on a tensor product of vector spaces $\bigotimes_{j=1}^m V_j \cong \C^N$, where $V_j \cong \C^{n_j}$.  We can then define the marginals $\pi_j(M)$, $j = 1, \hdots, m$ analogously to (\ref{eqn:marginal-def}) by taking the partial trace over all factors in the tensor product except $V_j$.  Write $\mu^{(j)}_1 \ge \hdots \ge \mu^{(j)}_{n_j}$ for the eigenvalues of $\pi_j(M)$.  We label the standard coordinates in $\R^N$ by multi-indices $k_1 \hdots k_m$ with $ 1 \le k_j \le n_j$, which we order lexicographically.  We then can ask, given $\lambda \in \R^N$ with $\lambda_{1 \hdots 1} \ge \hdots \ge \lambda_{n_1 \hdots n_m}$ and $M$ uniformly distributed on the $\U(N)$-orbit $\mathcal{O}_\lambda$, what is the joint distribution of $(\mu^{(1)}, \hdots, \mu^{(m)}) \in \R^{\sum_j n_j}$?  As in the two-body case, if all $\lambda_{k_1 \hdots k_m}$ are distinct and $\sum \lambda_{k_1 \hdots k_m} = 0$, then this distribution has a density $p_\lambda^\mathrm{dst}(\mu^{(1)}, \hdots, \mu^{(m)})$ with respect to Lebesgue measure on the subspace
$$\Sigma_0 = \Big \{ (\mu^{(1)}, \hdots, \mu^{(m)}) \in \R^{\sum_j n_j} \ \Big | \ \sum_{k=1}^{n_j} \mu^{(j)}_k = 0, \ j=1,\hdots,m \Big\}.$$
Similarly to before, we write $$\Sigma_0^\downarrow = \big \{ (\mu^{(1)}, \hdots, \mu^{(m)}) \in \Sigma_0 \ \big | \ \mu^{(j)}_1 \ge \hdots \ge \mu^{(j)}_{n_j}, \ j=1,\hdots,m \big \}.$$
A precisely analogous argument to the proof of Corollary \ref{cor:density-formula} then yields the following integral formula for $p_\lambda^\mathrm{dst}(\mu^{(1)}, \hdots, \mu^{(m)})$.

\begin{cor} \label{cor:density-multibody}
Take $\lambda \in \R^{N}$ with $\lambda_{1\hdots1} > \hdots > \lambda_{n_1 \hdots n_m}$ and $\sum \lambda_{j_1 \hdots j_m} = 0$. Then
\begin{multline} \label{eqn:density-multibody}
p_\lambda^\mathrm{dst}(\mu^{(1)}, \hdots, \mu^{(m)}) = C^\mathrm{dst}_{n_1, \hdots, n_m} \mathbbm{1}_{\Sigma_0^\downarrow}(\mu^{(1)}, \hdots, \mu^{(m)}) \frac{\prod_{j=1}^m \Delta_{n_j}(\mu^{(j)})}{\Delta_{N}(\lambda)} \int_{\Sigma_0} \\
\frac{ \prod_{j=1}^m \Delta_{n_j}(\xi^{(j)})}{ \Delta_{N}\big( [\sum_{j=1}^m \xi^{(j)}_{p_j}]_{p_1 \hdots p_m =1 \hdots 1}^{n_1 \hdots n_m} \big)} \, e^{-i \sum_{j=1}^m \sum_{k=1}^{n_j} \mu^{(j)}_k \xi^{(j)}_k } \det \big[ e^{i \lambda_{k_{1} \hdots k_{m}} \sum_{j=1}^m \xi^{(j)}_{p_j}} \big]_{k_1 \hdots k_m, p_1 \hdots p_m =1 \hdots 1}^{n_1 \hdots n_m} \\
dL(\xi^{(1)}, \hdots, \xi^{(m)}),
\end{multline}
where $dL$ is Lebesgue measure on $\Sigma_0$ and
\begin{equation} \label{eqn:kappa-def}
C^\mathrm{dst}_{n_1, \hdots, n_m} = \frac{(-i)^{\frac{1}{2}(N(N-1) - \sum_{j=1}^m n_j(n_j-1))}}{(2\pi)^{\sum_{j=1}^m n_j - m}} \frac{\prod_{l=1}^{N-1}l!}{ \prod_{j=1}^{m} \prod_{l=1}^{n_j-1} l!}.
\end{equation}
\end{cor}

\subsection{Bosons} \label{sec:formula-bos}
In the case of $k > 1$ indistinguishable bosons, we consider the space of symmetric tensors
$$ \mathcal{V}^\mathrm{bos}_{n | k} = \mathrm{Sym}^k \C^n, $$
which we model as degree-$k$ homogeneous polynomials in $n$ complex variables $x_1, \hdots, x_n$.  For two such polynomials $p,q$, we set
\begin{equation} \label{eqn:boson-inner-product}
(p,q) = p(\partial)q(x) \big |_{x=0}.
\end{equation}
The bilinear form $( \cdot, \cdot )$ defines a Hermitian inner product that makes $\mathcal{V}^\mathrm{bos}_{n | k}$ into a Hilbert space.  If $\alpha$ is an $n$-component multi-index, we write $|\alpha| = \sum_{i=1}^n \alpha_i$ and $\alpha! = \prod_{i=1}^n \alpha_i!$. Then the monomials
$$\frac{1}{\sqrt{\alpha!}}\prod_{i=1}^n x_i^{\alpha_i}$$
are an orthonormal basis of $\mathcal{V}^\mathrm{bos}_{n | k}$ as $\alpha$ runs over $n$-component multi-indices with $|\alpha| = k$.  As there are $N= {{n+k-1}\choose{k}}$ such multi-indices, we obtain isomorphisms
\begin{equation} \label{eqn:bos-isos}
\mathcal{V}^\mathrm{bos}_{n | k} \cong \C^{N}, \qquad \SU(\mathcal{V}^\mathrm{bos}_{n | k}) \cong \SU(N).
\end{equation}
Accordingly we also use multi-indices to label the standard coordinates in $\R^N$ or $\C^N$ and write, for example, $\lambda = ( \lambda_\alpha )_{|\alpha| = k} \in \R^{N}.$ We consider the multi-indices to be ordered lexicographically.

The usual action of $\SU(n)$ on the vector $(x_1, \hdots, x_n) \in \C^n$ gives the standard unitary representation of $\SU(n)$ on $\mathrm{Sym}^k \C^n$.  This representation extends to a representation on $(\C^n)^{\otimes k}$, which is obtained via the diagonal embedding
\begin{equation}
\label{eqn:bos-embed} \SU(n) \hookrightarrow \SU\big((\C^n)^{\otimes k}\big) \cong \SU(n^k).
\end{equation}
The corresponding embedding at the level of the Lie algebras is the map $\mathscr{S} : \mathfrak{su}(n) \hookrightarrow \mathfrak{su}(n^k)$ defined by
\begin{equation}
\label{eqn:sym-map-def}
\mathscr{S}(X) = \sum_{j = 1}^k \mathbbm{1}_{n^{j-1}} \otimes X \otimes \mathbbm{1}_{n^{k-j}}, \qquad X \in \mathfrak{su}(n).
\end{equation}

The quantum marginal problem for $k$ indistinguishable bosons is the special case of the problem studied in Section \ref{sec:projections-general} where $G = \SU(N)$ and $H$ is the image of $\SU(n)$ in $\SU(N)$ under this representation.

Concretely, using (\ref{eqn:bos-isos}) and (\ref{eqn:sym-map-def}) we may identify both $\mathfrak{su}(N)$ and $\mathfrak{su}(n)$ with subspaces of $\mathfrak{su}(n^k)$.  Multiplying by $i$, we regard elements of these Lie algebras as traceless {\it Hermitian} matrices.  This setup allows us to define the orthogonal projection $\phi: \mathrm{Her}(N)_0 \to \mathrm{Her}(n)_0$, where $\mathrm{Her}(m)_0$ indicates the space of $m$-by-$m$ Hermitian matrices with trace zero.  We take $\tot$ to be the space of $n$-by-$n$ traceless real diagonal matrices, which we identify with the space of vectors in $\R^n$ whose coordinates sum to zero.  Let $\tot_+ \subset \tot$ be the space of diagonal matrices with non-increasing entries down the diagonal and $\psi : \mathrm{Her}(n)_0 \to \tot_+$ be the diagonalization map.  Let $\mathcal{O}_\lambda \subset \mathrm{Her}(N)_0$ be an $\SU(N)$-orbit with orbital measure $\beta_\lambda$.  We then want to describe the pushforward measure $\psi_* \phi_* \beta_\lambda$.  To do so, we apply a direct specialization of Theorem \ref{thm:integral-formula-gen}.

If a matrix $M \in \mathrm{Her}(n)_0$ has eigenvectors $v_1, \hdots, v_n$ with respective eigenvalues $\mu_1 \ge \hdots \ge \mu_n$, then its image $\mathscr{S}(X)$, regarded as an operator on $\mathrm{Sym}^k \C^n$, has eigenvectors $\prod_{j = 1}^n v_j^{\alpha_j}$ with respective eigenvalues $\alpha \cdot \mu = \sum_{j = 1}^n \alpha_j \mu_j$, as $\alpha$ ranges over $n$-component multi-indices with $|\alpha| = k$.  Thus we find that the restriction of $\Delta_{\mathfrak{su}(N)}$ to $\tot$ is equal to
$$\Delta_{N}([\alpha \cdot \xi]_{|\alpha| = k}) \ = \prod_{\substack{\alpha < \beta \\ |\alpha| = |\beta| = k}} \big( \alpha \cdot \xi - \beta \cdot \xi \big), \qquad \xi \in \tot.$$
It is easy to check directly that the above expression does not vanish uniformly on $\tot$, and that the multiset of vectors $\Phi^+_{\gog/\goh}$ spans $\tot$. With these observations, Theorem \ref{thm:integral-formula-gen} gives the following:

\begin{cor} \label{cor:density-formula-bos}
Take $\lambda \in \R^{N}$ with all $\lambda_\alpha$ distinct. Then the pushforward measure $\psi_* \phi_* \beta_\lambda$ has a density with respect to Lebesgue measure on $\tot_+$, which is given by
\begin{equation} \label{eqn:density-formula-bos}
p^{\mathrm{bos}}_\lambda(\mu) = C^\mathrm{bos}_{n | k} \, \frac{\Delta_n(\mu)}{\Delta_{N}(\lambda)} \int_{\tot}
\frac{ \Delta_n(\xi)}{ \Delta_{N}([\alpha \cdot \xi]_{|\alpha| = k})} \, e^{-i \sum \mu_j \xi_j} \det \big[ e^{i \lambda_{\alpha} (\beta \cdot \xi) } \big]_{|\alpha| , |\beta| = k} \, d\xi,
\end{equation}
where $d\xi$ is Lebesgue measure on $\tot \cong \{ x \in \R^n \ | \ \sum x_j = 0 \}$ and
\begin{equation} \label{eqn:kappa-def-bos}
C^\mathrm{bos}_{n | k} = \frac{(-i)^{\frac{1}{2}(N(N-1) - n(n-1))}}{(2\pi)^{n-1}} \prod_{l=n}^{N-1}l!.
\end{equation}
\end{cor}

\subsection{Fermions} \label{sec:formula-fer}
In the case of $k$ indistinguishable fermions, we consider the space of alternating tensors
$$ \mathcal{V}^\mathrm{fer}_{n | k} = \wedge^k \C^n, $$
which we model as $k$-vectors
$$v_1 \wedge \cdots \wedge v_k, \qquad v_1, \hdots, v_k \in  \C^n.$$
We make $\mathcal{V}^\mathrm{fer}_{n | k}$ into a Hilbert space by choosing as an orthonormal basis the $k$-vectors
$$e_{i_1} \wedge \cdots \wedge e_{i_k}, \qquad 1 \le i_1 < \hdots < i_k \le n,$$
where $e_j$ is the $j$th standard basis vector of $\C^n$.  There are $K = {{n} \choose {k}}$ such $k$-vectors labeled by $n$-component multi-indices $\alpha$ with $k$ entries equal to 1 and all other entries equal to zero.  We write $\mathcal{I}$ for the set of such multi-indices, which we again order lexicographically by their entries.  We use $\mathcal{I}$ to label the standard coordinates in $\R^K$ or $\C^K$, writing e.g.~$\lambda = ( \lambda_\alpha )_{\alpha \in \mathcal{I}} \in \R^{K}.$ As in the bosonic case, we then obtain isomorphisms
$$\mathcal{V}^\mathrm{fer}_{n | k} \cong \C^{K}, \qquad \SU(\mathcal{V}^\mathrm{fer}_{n | k}) \cong \SU(K),$$
as well as homomorphisms
\begin{equation} \label{eqn:fer-embed}
\SU(n) \to \SU(K), \qquad \mathfrak{su}(n) \to \mathfrak{su}(K),
\end{equation}
where $\SU(n)$ acts on $\wedge^k \C^n$ by acting in the usual way on each of the vectors in the wedge product $v_1 \wedge \cdots \wedge v_k$. Again, this representation extends to representation on $(\C^n)^{\otimes k}$ via the embeddings (\ref{eqn:bos-embed}), (\ref{eqn:sym-map-def}).

The quantum marginal problem for $k$ indistinguishable fermions is the special case of the problem studied in Section \ref{sec:projections-general} where $G = \SU(K)$ and $H$ is the image of $\SU(n)$ in $\SU(K)$ under (\ref{eqn:fer-embed}).  Note that if $k = 1$ or $n-1$ then $K = n$.  In these cases $G = H$, so that the quantum marginal problem is trivial (see Remark \ref{rem:unfaithful}).  Moreover, if $k = n$ then $K = 1$ and $\mathcal{V}^\mathrm{fer}_{n | n} \cong \C$. In this case the maps in (\ref{eqn:fer-embed}) are constant, and again the problem is trivial.  Finally, if $k > n$ then $\mathcal{V}^\mathrm{fer}_{n | k} = \{ 0 \}$, so that the quantum marginal problem is both trivial and physically meaningless.  Accordingly, we assume from now on that $1 < k < n-1$.

As before, we work over the spaces of traceless Hermitian matrices $\mathrm{Her}(K)_0$ and $\mathrm{Her}(n)_0$, which we regard as subspaces of $\mathrm{Her}(n^k)_0$.  We reuse the same notation, writing $\phi: \mathrm{Her}(K)_0 \to \mathrm{Her}(n)_0$ for the orthogonal projection and $\beta_\lambda$ for the invariant probability measure on an $\SU(K)$-orbit $\mathcal{O}_\lambda \subset \mathrm{Her}(K)_0$.  Exactly as in the bosonic case, we write $\tot$ for the space of $n$-by-$n$ real diagonal matrices, $\tot_+ \subset \tot$ for the space of diagonal matrices with non-increasing entries down the diagonal, and $\psi : \mathrm{Her}(n)_0 \to \tot_+$ for the diagonalization map.  Again, we will describe the pushforward measure $\psi_* \phi_* \beta_\lambda$ by a specialization of Theorem \ref{thm:integral-formula-gen}.

To each multi-index $\alpha \in \mathcal{I}$, we associate the $k$-tuple
$$i(\alpha) = \big(i(\alpha)_1, \, \hdots, \, i(\alpha)_k \big) \in \{ 1, \hdots, n \}^k$$
such that $i(\alpha)_j$ is equal to the position of the $j$th nonzero entry of $\alpha$.  If a matrix $M \in \mathrm{Her}(n)_0$ has eigenvectors $v_1, \hdots, v_n$ with respective eigenvalues $\mu_1 \ge \hdots \ge \mu_n$, then its image in $\mathrm{Her}(K)_0$ under (\ref{eqn:fer-embed}), regarded as an operator on $\wedge^k \C^n$, has eigenvectors
$v_{i(\alpha)_{1}} \wedge \cdots \wedge v_{i(\alpha)_{k}}$
with respective eigenvalues $\alpha \cdot \mu$, for $\alpha \in \mathcal{I}$.  We thus find that
$$\Delta_\gog(\xi) = \Delta_{K}([\alpha \cdot \xi]_{\alpha \in \mathcal{I}}), \qquad \xi \in \tot.$$
Again it is easy to check that the above expression does not vanish uniformly and that $\Phi^+_{\gog/\goh}$ spans $\tot$.  From Theorem \ref{thm:integral-formula-gen}, we then find:

\begin{cor} \label{cor:density-formula-fer}
Take $\lambda \in \R^{K}$ with all $\lambda_\alpha$ distinct for $\alpha \in \mathcal{I}$. Then the pushforward measure $\psi_* \phi_* \beta_\lambda$ has a density with respect to Lebesgue measure on $\tot_+$, which is given by
\begin{equation} \label{eqn:density-formula-fer}
p^{\mathrm{fer}}_\lambda(\mu) = C^\mathrm{fer}_{n | k} \, \frac{\Delta_n(\mu)}{\Delta_{K}(\lambda)} \int_{\tot}
\frac{ \Delta_n(\xi)}{ \Delta_{K}([\alpha \cdot \xi]_{\alpha \in \mathcal{I}})} \, e^{-i \sum \mu_j \xi_j} \det \big[ e^{i \lambda_{\alpha} (\beta \cdot \xi) } \big]_{\alpha , \beta \in \mathcal{I}} \, d\xi,
\end{equation}
where $d\xi$ is Lebesgue measure on $\tot \cong \{ x \in \R^n \ | \ \sum x_j = 0 \}$ and
\begin{equation} \label{eqn:kappa-def-bos}
C^\mathrm{fer}_{n | k} = \frac{(-i)^{\frac{1}{2}(K(K-1) - n(n-1))}}{(2\pi)^{n-1}} \frac{\prod_{l=1}^{K-1}l!}{\prod_{l=1}^{n-1}l!}.
\end{equation}
\end{cor}

\section{Properties of the densities} \label{sec:properties}

In this section we establish some basic properties of the density $p_\lambda$ described in Theorem \ref{thm:integral-formula-gen}.  We start with the following fact, which is an immediate consequence of two celebrated results in symplectic geometry: the Duistermaat--Heckman theorem \cite{DH} and the Kirwan convexity theorem \cite{Kirwan-convexity}.

\begin{prop} \label{prop:piecewise-poly}
Under the assumptions of Theorem \ref{thm:integral-formula-gen}, the density $p_\lambda$ is a piecewise polynomial function supported on a convex polytope in $\tot_+$, and its non-analyticities are contained in a finite collection of hyperplanes.  The local polynomial expressions of $p_\lambda$ have degree at most $|\Phi_\gog^+| + |\Phi_\goh^+| - r$, and all of them are divisible by $\Delta_\goh^2$.
\end{prop}

\begin{proof}
We first recall some basic facts from symplectic geometry. The coadjoint orbit $\mathcal{O}_\lambda$ carries a canonical $G$-invariant symplectic form, the Kostant--Kirillov--Souriau form \cite{Kirillov-lectures}.  The Liouville measure of this symplectic form is equal to
\begin{equation} \label{eqn:liouville-volume}
\frac{\Delta_\gog(\lambda)}{\Delta_\gog(\rho_\gog)} \beta_\lambda.
\end{equation}
After identifying the Lie algebras $\gog$ and $\goh$ with their duals, the projection $\phi: \gog \to \goh$ is a moment map for the action of $H$ on $\mathcal{O}_\lambda$.  The {\it Duistermaat--Heckman measure} associated to this group action is defined as the pushforward of the Liouville measure (\ref{eqn:liouville-volume}) by the moment map $\phi$; up to a constant, this is just the measure $\phi_* \beta_\lambda$ studied in Section \ref{sec:formulae}.  In the setting of $H$ acting on $\mathcal{O}_\lambda$, the Duistermaat--Heckman theorem \cite{DH} tells us that in a neighborhood of any regular value\footnote{Recall that a {\it regular value} of $\phi$ is a point $x \in \goh$ such that at every point $y$ in the preimage $\phi^{-1}(x)$, the differential map $d_y \phi: T_y \mathcal{O}_\lambda \to T_x \goh$ is surjective.} of $\phi$, the density $q_\lambda$ of $\phi_* \beta_\lambda$ is locally equal to a polynomial of degree at most
$$\frac{1}{2}(\dim \mathcal{O}_\lambda - \dim \goh - r) = |\Phi_\gog^+| - |\Phi_\goh^+| - r.$$
Additionally, if $x \in \goh$ is {\it not} a regular value of $\phi$, then its radial component $\psi(x) \in \tot_+$ satisfies at least one of a finite collection of linear equations depending on $\lambda$.  The Kirwan convexity theorem \cite{Kirwan-convexity} tells us that the restriction of $q_\lambda$ to $\tot_+$ is supported on a convex polytope.\footnote{A theorem of Berenstein--Sjamaar \cite{BerensteinSjamaar} gives a list of inequalities that cut out the facets of this supporting polytope, generalizing the results of Klyachko \cite{Klyachko-marginals} on the quantum marginal problem.}

With the above standard results in hand, all of the desired statements follow from the facts that $p_\lambda(x) = \kappa_\goh \Delta_\goh(x)^2 q_\lambda(x)$ for $x \in \tot_+$, and that $\deg \Delta_\goh = |\Phi_\goh^+|$.
\end{proof}

Figure 1 shows an example of the joint spectral density of the marginals of two distinguishable particles, with the local polynomial regions delineated, together with a numerically generated histogram.

For the special cases considered in Sections \ref{sec:formula-dist}, \ref{sec:formula-bos} and \ref{sec:formula-fer}, we find:

\begin{cor} \label{cor:degrees-QMprobs}
In the settings of Corollaries \ref{cor:density-formula}, \ref{cor:density-formula-bos} and \ref{cor:density-formula-fer} respectively, the local polynomial expressions of the densities $p_\lambda^{\mathrm{dst}}$, $p_\lambda^{\mathrm{bos}}$ and  $p_\lambda^{\mathrm{fer}}$ have degrees at most
$$\frac{1}{2}\big( mn(mn-1) + (m-1)(m-2) + (n-1)(n-2) \big),$$
$$\frac{1}{2}\big( N(N-1) + (n-1)(n-2) \big),$$
and
$$\frac{1}{2}\big( K(K-1) + (n-1)(n-2) \big).$$
\end{cor}

\begin{center}
\includegraphics[width=165px]{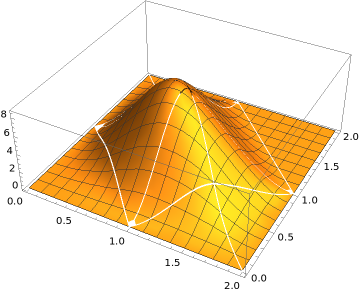}
\includegraphics[width=185px]{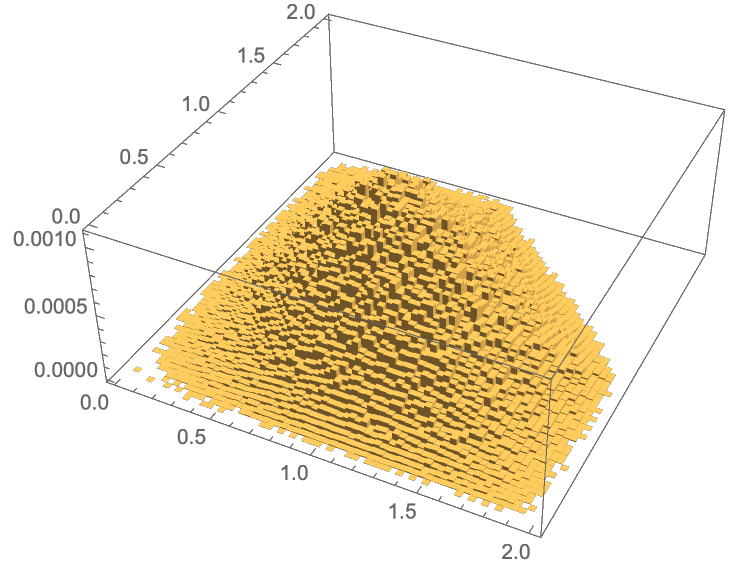}

\footnotesize
{\it Figure 1 (courtesy of Jean-Bernard Zuber \cite{Zuber-QMfigs}):} Joint spectral density of the marginals of two distinguishable particles with $m = n = 2$ and $\lambda = (\frac{3}{2}, \frac{1}{2}, -\frac{1}{2}. -\frac{3}{2})$.  The density is shown in the $(\mu_1, \nu_1)$ plane. \textbf{Left:} Plot of an analytic expression for the density derived from the formula (\ref{eqn:density-formula}), showing the 7 singular hyperplanes (in this case, lines).  The equations for the singular lines are $\mu_1 \in \{1, 2\}$, $\nu_1 \in \{1, 2\}$, and $\mu_1 \pm \nu_1 \in \{1,2,3\}$. \textbf{Right:} Histogram of spectra of the single-particle marginals of $10^6$ uniform random samples from the orbit $\mathcal{O}_\lambda$.
\normalsize
\end{center}
\medskip

Although $p_\lambda$ is not globally smooth, in general it has a minimum number of continuous derivatives even at points where it is not analytic.  Next we prove a lower bound on the degree of differentiability of $p_\lambda$ at each point of $\tot_+$ and give an explicit expression for this bound in the case of the quantum marginal problem studied in Section \ref{sec:formula-dist}.  The proof, based on a method used by Zuber to study the density of the randomized Horn's problem \cite{Z}, does not require any symplectic geometry.  We will need the following elementary fact of Fourier analysis.

\begin{lem} \label{lem:FT-decay-regularity}
Let $f \in L^1(\R^d)$.  If $Q \cdot \mathscr{F}[f] \in L^1(\R^d)$ for all polynomials $Q$ of degree $k$, then $f \in C^k(\R^d)$.
\end{lem}

\begin{proof}
Suppose that $Q \cdot \mathscr{F}[f] \in L^1(\R^d)$ for all polynomials $Q$ of degree $k$.  It suffices to show that $D f \in C^0(\R^d)$ for $D$ an arbitrary homogeneous differential operator of order $k$.  For $v \in \R^d$, write $\partial_v f(x) = \frac{d}{dt} f(x + tv)  |_{t=0}$.  We can decompose $D$ as a product of first-order differential operators $D = \partial_{v_1} \cdots \partial_{v_k}$ for some $v_1, \hdots, v_k \in \R^d$.  Then since the function
$$\bigg( \prod_{j=1}^k \langle -i \xi, v_j \rangle \bigg ) \mathscr{F}[f](\xi)$$
is in $L^1(\R^d)$, its inverse Fourier transform is uniformly continuous and equals
\begin{align*}
\mathscr{F}^{-1} \left[ \bigg( \prod_{j=1}^k \langle -i \xi, v_j \rangle \bigg ) \mathscr{F}[f](\xi) \right](x) &= \frac{1}{(2 \pi)^{d/2}} \int_{\R^d} \bigg( \prod_{j=1}^k \langle -i \xi, v_j \rangle \bigg ) \mathscr{F}[f](\xi) \, e^{-i\langle \xi, x \rangle} d \xi \\
&= \frac{1}{(2 \pi)^{d/2}} \int_{\R^d} \mathscr{F}[f](\xi) \, D e^{-i\langle \xi, x \rangle} d \xi,
\end{align*}
where the operator $D$ acts in the $x$ variables.  The functions $\mathscr{F}[f](\xi) \, e^{-i\langle \xi, x \rangle}$ and $\mathscr{F}[f](\xi) \, D e^{-i\langle \xi, x \rangle}$ are both uniformly continuous in $x$ and $\xi$, since $f \in L^1(\R^d)$.  Thus we may pull $D$ out of the integral to obtain
\begin{align*}
\mathscr{F}^{-1} \left[ \bigg( \prod_{j=1}^k \langle -i \xi, v_j \rangle \bigg ) \mathscr{F}[f](\xi) \right](x) &= \frac{1}{(2 \pi)^{d/2}} D \int_{\R^d} \mathscr{F}[f](\xi) \, e^{-i\langle \xi, x \rangle} d \xi \\
&= D \mathscr{F}^{-1}[ \mathscr{F}[f]](x) = Df(x).
\end{align*}
Therefore $Df(x)$ exists and is uniformly continuous, which completes the proof.
\end{proof}

We now show a result on the differentiability of $p_\lambda$.  The essential idea is to use the lemma above and show that $Q \cdot \mathscr{F}_\tot [p_\lambda] \in L^1(\tot)$ whenever $Q$ is a polynomial of suitably low degree.  In order to show integrability, we partition $\tot$ into finitely many subsets and consider the integral over each subset separately.  Below in Remark \ref{rem:spline-regularity} we sketch an alternate proof, which is superficially shorter but relies on abstract constructions from representation theory.

\begin{thm} \label{thm:density-regularity}
In the setting of Theorem \ref{thm:integral-formula-gen}, let $\ell$ be the largest integer such that all subsets of $\Phi^+_{\gog/\goh}$ obtained by deleting $\ell+1$ vectors (counted with multiplicity) span $\tot$.  Then $p_\lambda \in C^\ell(\tot_+)$.
\end{thm}

\begin{remark} \label{rem:negative-ell}
Clearly $\ell \ge -1$, since $\Phi^+_{\gog/\goh}$ is assumed to span $\tot$.  In some cases $\ell = -1$ and $p_\lambda$ is discontinuous. For example, this occurs when $G = \SU(2) \times \SU(2)$ and $H$ is the diagonal subgroup, in which case $\tot$ is 1-dimensional and $\Phi^+_{\gog / \goh}$ contains only a single vector (see \cite{CZ1}).
\end{remark}

\begin{proof}[Proof of Theorem \ref{thm:density-regularity}]
We assume that $\ell \ge 0$, as otherwise there is nothing to prove.  Let $Q$ be an arbitrary polynomial on $\tot$ of degree $\ell$.  From Corollary \ref{cor:pushfwd-FT}, we have
\begin{equation} \label{eqn:p-as-FT}
p_\lambda(x) = \frac{\Delta_\goh(x)}{\Delta_\goh(\rho_\goh)} \mathscr{F}_\tot^{-1}\big[ \Delta_\goh( i \bullet ) \, \mathcal{H}_G(\lambda, i \bullet ) \big] (x), \qquad x \in \tot_+.
\end{equation}
By Lemma \ref{lem:FT-decay-regularity}, it thus suffices to show that $Q(\bullet) \Delta_\goh(\bullet) \mathcal{H}_G(\lambda, i \bullet ) \in L^1(\tot)$.  From the Harish-Chandra formula (\ref{eqn:hc-integral}), we have
\begin{align*}
\Delta_\goh(\xi) \mathcal{H}_G(\lambda, i \xi) &= \frac{\Delta_\gog(\rho_\gog) \Delta_\goh(\xi)}{\Delta_\gog(\lambda) \Delta_\gog(i \xi)} \sum_{w \in W_\gog} \epsilon(w) \, e^{i \langle w(\lambda), \xi \rangle} \\
&\propto \frac{1}{\Delta_{\gog / \goh}(\xi)} \sum_{w \in W_\gog} \epsilon(w) \, e^{i \langle w(\lambda), \xi \rangle}, \qquad \xi \in \tot.
\end{align*}
Write $F(\xi) = \Delta_{\gog / \goh}(\xi)^{-1} \sum_{w \in W_\gog} \epsilon(w) \, e^{i \langle w(\lambda), \xi \rangle}$ for the function appearing in the final line of the display above.  As discussed above in Remark \ref{rem:integral-convergence}, $F$ is ostensibly singular on the hyperplanes where $\Delta_{\gog / \goh}$ vanishes, but these singularities are removable, as the sum over the Weyl group vanishes to at least the same order at all such points.  In fact, the relation (\ref{eqn:p-as-FT}) states that $F$ is proportional to the Fourier transform of $\Delta_\goh^{-1} p_\lambda = \kappa_\goh \Delta_\goh q_\lambda$, which is a compactly supported, integrable function.  Therefore $F$ is uniformly continuous, decays to 0 at infinity on $\tot$, and extends to a holomorphic function on the complexification $\tot_\C$.

In order to show that $QF \in L^1(\tot)$, we will introduce a finite partition of $\tot$ and show that $|QF|$ is integrable on each set in the partition.  For each $\alpha \in \Phi_{\gog / \goh}^+$, let $H_{\alpha} = \{ x \in \tot \ | \ \langle \alpha, x \rangle = 0 \}$ be the hyperplane normal to $\alpha$.  Then
\begin{equation} \label{eqn:hyperplane-dist}
| \langle \alpha, x \rangle | = |\alpha| \, \mathrm{dist}(x, H_\alpha),
\end{equation}
where $\mathrm{dist}(x, H_\alpha)$ is the Euclidean distance.  The complement of the union of hyperplanes $$\bigcup_{\alpha \in \Phi^+_{\gog / \goh}} H_\alpha$$ consists of finitely many connected components, each of which is an open convex polyhedral cone.  Let $\mathcal{C}$ be an arbitrary one of these cones, $\overline{\mathcal{C}}$ its closure, and $\Sigma$ the set of faces of $\overline{ \mathcal{C} }$.  Since $\Phi^+_{\gog / \goh}$ spans $\tot$, the cone $\overline{\mathcal{C}}$ is salient (i.e.,~$\overline{\mathcal{C}} \setminus \{0\}$ is contained in some open half-space).  Each face $\sigma \in \Sigma$ is itself a salient convex polyhedral cone and is the intersection of $\overline{\mathcal{C}}$ with some number of hyperplanes $H_\alpha$.  For $\sigma \in \Sigma$, let $V(\sigma) = \{ \alpha \in \Phi^+_{\gog / \goh} \ | \ \sigma \subset H_\alpha \}$, and define
\begin{align*}
N_\sigma = \big \{ x \in \overline{\mathcal{C}} \ \big| \ &\mathrm{dist}(x, H_\alpha) \le |\alpha|^{-1} \quad \forall \alpha \in V(\sigma), \\
& \mathrm{dist}(x, H_\alpha) > |\alpha|^{-1} \quad \forall \alpha \not \in V(\sigma) \big \}.
\end{align*}

To show that $QF \in L^1(\tot)$, it clearly suffices to show that $\int_{\overline{\mathcal{C}}} |Q(x) F(x)| dx < \infty$, and  since the sets $N_\sigma$ form a finite partition of $\overline{\mathcal{C}}$, we need only show that the integral over each $N_\sigma$ is finite.  We first show some preliminary estimates on $F$.  Set
$$\Delta_\sigma(x) = \prod_{\alpha \not \in V(\sigma)} \langle \alpha, x \rangle.$$
Then by (\ref{eqn:hyperplane-dist}), $| \Delta_\sigma(x) | > 1$ on $N_\sigma$.  Write
$$ F(x) = \frac{f_\sigma(x)}{\Delta_\sigma(x)},$$
where
\begin{equation} \label{eqn:fsigma-def}
f_\sigma(x) = \frac{\sum_{w \in W_\gog} \epsilon(w) \, e^{i \langle w(\lambda), x \rangle}}{\prod_{\alpha \in V(\sigma)} \langle \alpha, x \rangle}.
\end{equation}
A direct computation with L'H\^opital's rule, differentiating the numerator and denominator of (\ref{eqn:fsigma-def}) each $|V(\sigma)|$ times, shows that the restriction of $f_\sigma$ to $\sigma$ is a trigonometric polynomial and is therefore bounded.  Thus we have the inequality
$$ |F(x)| < \frac{C}{| \Delta_\sigma(x) |}, \qquad x  \in \sigma.$$
(Here and below, we use $C$ to indicate a general positive constant whose exact value may differ between expressions.)  Since $F$ is uniformly continuous and all points of $N_\sigma$ lie within a bounded distance from $\sigma$, we then obtain the estimate
\begin{equation} \label{eqn:F-bound2}
|F(x)| < \frac{C}{| \Delta_\sigma(x) |}, \qquad x  \in N_\sigma.
\end{equation}
Since $|Q(x)| < C |x|^\ell$, it therefore suffices to show
\begin{equation} \label{eqn:estimate-simplified1}
\int_{N_\sigma} \frac{|x|^\ell}{|\Delta_\sigma(x)|} dx < \infty
\end{equation}
 for each $\sigma \in \Sigma$.

The only bounded face is $\sigma = \{0\}$, in which case $\Delta_\sigma \equiv 1$ and (\ref{eqn:estimate-simplified1}) obviously holds.  If $\sigma$ is any other (unbounded) face, we can reduce the estimate (\ref{eqn:estimate-simplified1}) to bounding an integral over $\sigma \cap N_\sigma$ rather than $N_\sigma$, as follows.  Let $P : \tot \to \mathrm{span}(\sigma)$ be the orthogonal projection onto the linear span of $\sigma$.  Then for $x \in N_\sigma$ we have $P(x) \in \sigma$, and the triangle inequality gives
$$\min_{\alpha \in \Phi^+_{\gog / \goh}} |\alpha|^{-1} < |x| \le |P(x)| + |x - P(x)| < |P(x)| + \sum_{\alpha \in V(\sigma)} |\alpha|^{-1},$$
which implies $|x|^\ell < C |P(x)|^\ell.$  Similarly,  for $\alpha \in \Phi^+_{\gog \ \goh} \setminus V(\sigma)$, $| \langle \alpha, x \rangle | > 1$ and
$$ | \langle \alpha, P(x) \rangle | \le |\langle \alpha, x \rangle| + |\langle \alpha, x - P(x) \rangle| < |\langle \alpha, x \rangle| + |\alpha| \sum_{\beta \in V(\sigma)} |\beta|^{-1},$$
which implies $|\Delta_\sigma(P(x))| < C |\Delta_\sigma(x)|.$  Therefore
$$\int_{N_\sigma} \frac{|x|^\ell}{|\Delta_\sigma(x)|} dx < C \int_{N_\sigma} \frac{|P(x)|^\ell}{|\Delta_\sigma\big(P(x) \big)|} dx < C \int_{\sigma \cap N_\sigma} \frac{|y|^\ell}{|\Delta_\sigma(y)|} dy,$$
where $dy$ is Lebesgue measure on the span of $\sigma$.

Before showing that this last integral is finite, we make yet another simplification to the integrand.  For $x \in \tot,$ define
$$[ x ]_\ell =\bigg( \sum_{\{\alpha_1, \hdots, \alpha_\ell \} \subset \Phi^+_{\gog / \goh}} \prod_{j = 1}^\ell |\langle \alpha_j, x \rangle | \bigg)^{1/\ell}$$
when $\ell \ge 1$, and $[x]_0 = 1$.  For $\ell \ge 1$ it is easy to check by direct calculation that $[ \bullet ]_\ell$ satisfies the properties of a norm, and since all norms on a finite-dimensional vector space are equivalent, we have $|x| < C[x]_\ell$.  Additionally, since $\sigma$ lies in the orthogonal complement of all $\alpha \in V(\sigma)$, for $x \in \sigma$ we have
\begin{equation} \label{eqn:lbracket-alt}
[ x ]_\ell^\ell = \sum_{\substack{J \subset \Phi^+_{\gog / \goh} \setminus V(\sigma) \\ |J| = \ell}} \bigg( \prod_{\alpha \in J} |\langle P(\alpha), x \rangle | \bigg).
\end{equation}
When $\ell = 0$, the empty set is the only subset $J \subset \Phi^+_{\gog / \goh} \setminus V(\sigma)$ with $|J| = \ell$, but (\ref{eqn:lbracket-alt}) still holds in this case by the convention that a product over the empty set is equal to 1.  Therefore, regardless of whether $\ell$ is strictly positive, we have
\begin{align*}
\int_{\sigma \cap N_\sigma} \frac{|y|^\ell}{|\Delta_\sigma(y)|} dy &< C \int_{\sigma \cap N_\sigma} \frac{[y]_\ell^\ell}{|\Delta_\sigma(y)|} dy \\ &= C \sum_{\substack{J \subset \Phi^+_{\gog / \goh} \setminus V(\sigma) \\ |J| = \ell}} \int_{\sigma \cap N_\sigma} \frac{dy}{\prod_{\alpha \not \in (V(\sigma) \cup J)} |\langle P(\alpha), y \rangle |}.
\end{align*}
Fix an arbitrary subset $J \subset \Phi^+_{\gog / \goh} \setminus V(\sigma)$ with $|J| = \ell$, and let $I = \Phi^+_{\gog / \goh} \setminus (V(\sigma) \cup J).$  (Elements of $I$ and $J$, like those of $\Phi^+_{\gog / \goh}$, are counted with multiplicity.) To complete the proof, we now only need to show
\begin{equation} \label{eqn:estimate-simplified2}
\int_{\sigma \cap N_\sigma} \frac{dy}{\prod_{\alpha \in I} |\langle P(\alpha), y \rangle |} < \infty.
\end{equation}

We first make two observations about the set $I.$  Let $d = \dim \sigma$.  The first observation is that $I$ contains at least $d + 1$ elements, and if $I' = I \setminus \{\alpha\}$ is any subset of $I$ obtained by deleting a single element $\alpha$, then
$$\mathrm{span} \{ P(\beta) \ | \ \beta \in I' \} = \mathrm{span}(\sigma).$$  If this were not so, then the set $\Phi^+_{\gog / \goh} \setminus (J \cup \{\alpha \})$ would have a nontrivial orthogonal complement in $\mathrm{span}(\sigma)$, contradicting the assumption that all subsets of $\Phi^+_{\gog / \goh}$ obtained by deleting $\ell + 1$ elements span $\tot$.

The second observation is that we can assume without loss of generality that $\langle P(\alpha), y \rangle \ge 0$ for all $y \in \sigma$, since each $\alpha \in I$ takes only a single sign on $\sigma \cap N_\sigma$ and the integrand in (\ref{eqn:estimate-simplified2}) depends only on $|\langle P(\alpha), y \rangle |$.  With this assumption, the vectors $P(\alpha)$, $\alpha \in I$ belong to the dual cone $\sigma^*$ of $\sigma$ in $\mathrm{span}(\sigma)$.

We claim that there is an open half-space containing both $\sigma \setminus \{0\}$ and $\sigma^* \setminus \{0\}$.  To see this, notice that the intersection $\sigma \cap \sigma^*$ must have nonempty relative interior in $\mathrm{span}(\sigma)$; otherwise by the hyperplane separation theorem there would be some $x \in \mathrm{span}(\sigma)$ with $\langle x, y \rangle  > 0$ for all $y$ in the relative interior of $\sigma$ and $\langle x, z \rangle  < 0$ for all $z$ in the relative interior of $\sigma^*$.  But then we would have $x \in \sigma^*$, implying $\langle x, x \rangle \le 0$, a contradiction.  If we choose some point $z$ in the relative interior of $\sigma \cap \sigma^*$, then both $\sigma \setminus \{0\}$ and $\sigma^* \setminus \{0\}$ belong to the open half-space $\{ x \in \mathrm{span}(\sigma) \ | \ \langle z, x \rangle > 0 \}$.  Thus we can find a basis $v_1, \hdots, v_d$ of $\mathrm{span}(\sigma)$ such that the vectors $P(\alpha)$ for $\alpha \in I$, along with the entire set $\sigma \cap N_\sigma$, lie on the interior of the positive cone generated by the $v_j$.  In particular, for $\alpha \in I$ we can write
$$P(\alpha) = \sum_{j = 1}^d c^\alpha_j v_j$$
where the coefficients $c^\alpha_j$ are all strictly positive.

Finally, we show the inequality (\ref{eqn:estimate-simplified2}).  Set $\varepsilon = \min_{\alpha \in \Phi^+_{\gog / \goh}} |\alpha|^{-1}.$  Then all points of $\sigma \cap N_\sigma$ lie at least a distance $\varepsilon$ from the boundary of the positive cone generated by the $v_j$, so that
\begin{align*}
\int_{\sigma \cap N_\sigma} \frac{dy}{\prod_{\alpha \in I} |\langle P(\alpha), y \rangle |} &< C \int_\varepsilon^\infty \cdots \int_\varepsilon^\infty \frac{dx_1 \cdots dx_d}{\prod_{\alpha \in I} \sum_{j=1}^d c^\alpha_j x_j} \\
&< C \int_\varepsilon^\infty \cdots \int_\varepsilon^\infty \frac{dx_1 \cdots dx_d}{\left( \prod_{j = 1}^d x_j \right) \left( \sum_{j=1}^d c_j x_j \right)}
\end{align*}
for some $c_1, \hdots, c_d > 0$.  Then, since
$$\bigg( \prod_{j = 1}^d x_j \bigg) \bigg( \sum_{j=1}^d c_j x_j \bigg) = \prod_{j = 1}^d \left[ x_j \bigg( \sum_{k=1}^d c_k x_k \bigg)^{1/d} \right] > C  \prod_{j = 1}^d x_j^{1 + \frac{1}{d}},$$
we find
$$\int_\varepsilon^\infty \cdots \int_\varepsilon^\infty \frac{dx_1 \cdots dx_d}{\left( \prod_{j = 1}^d x_j \right) \left( \sum_{j=1}^d c_j x_j \right)} < C \int_\varepsilon^\infty \cdots \int_\varepsilon^\infty \frac{dx_1 \cdots dx_d}{\prod_{j = 1}^d x_j^{1 + \frac{1}{d}}}.$$
This last integral is finite, which completes the proof.
\end{proof}

The following lemma gives an alternate characterization of the number $\ell$ in terms of the polynomial $\Delta_{\gog / \goh}$.

\begin{lem} \label{lem:Delta-growth}
Let $\ell$ be as in Theorem \ref{thm:density-regularity}.  The restriction of $\Delta_{\gog / \goh}$ to any $d$-dimensional affine subspace of $\tot$ is either uniformly zero or has degree at least $\ell + d + 1$.  Additionally, there exists an affine line on which the minimum degree of $\ell + 2$ is realized.
\end{lem}

\begin{proof}
We first consider affine lines.  Choose $y, \eta \in \tot$, and write $$\Delta_y^\eta(t) = \Delta_{\gog / \goh}(ty + \eta), \qquad t \in \R$$ for the restriction of $\Delta_{\gog / \goh}$ to the line through $\eta$ with slope $y$.  Then $\Delta_y^\eta(t)$ is a polynomial in $t$, and if it is not uniformly zero then its degree is equal to the number of $\alpha \in \Phi^+_{\gog / \goh}$, counted with multiplicity, for which $\langle \alpha, y \rangle \ne 0$.  This degree must be at least $\ell + 2$, since the set $\{ \alpha \in \Phi^+_{\gog / \goh} \ | \ \langle \alpha, y \rangle = 0 \}$ is contained in the hyperplane orthogonal to $y$ and therefore does not span $\tot$.  Conversely, the definition of $\ell$ implies that there is some subset $E \subset \Phi^+_{\gog / \goh}$ with $|E| = \ell + 2$ (counted with multiplicity) such that $\Phi^+_{\gog / \goh} \setminus E$ does not span $\tot$.  Let $y$ lie in the orthogonal complement of $\Phi^+_{\gog / \goh} \setminus E$ and let $\eta$ be any point with $\Delta_{\gog / \goh}(\eta) \ne 0$.  Then $\Delta_y^\eta$ has degree $\ell + 2$.

Similarly, if $y_1, \hdots, y_d \in \tot$ are linearly independent, we write
$$ \Delta_{y_1, \hdots, y_d}^\eta(t_1, \hdots, t_d) = \Delta_{\gog / \goh}(t_1 y_1 + \cdots + t_d y_d + \eta), \qquad t_1, \hdots, t_d \in \R$$
for the restriction of $\Delta_{\gog / \goh}$ to the $d$-dimensional affine subspace $\mathrm{span}\{y_1, \hdots, y_d \} + \eta$.  If $\Delta_{y_1, \hdots, y_d}^\eta$ is not uniformly zero, then its degree is equal to the number of $\alpha \in \Phi^+_{\gog / \goh}$, counted with multiplicity, satisfying $\langle \alpha, y_j \rangle \ne 0$ for at least one $y_j$.  There must be at least $\ell + d + 1$ such $\alpha$, since all other vectors in $\Phi^+_{\gog / \goh}$ are contained in the orthogonal complement of $\mathrm{span}\{y_1, \hdots, y_d \}$.  Therefore $\deg  \Delta_{y_1, \hdots, y_d}^\eta \ge \ell + d + 1.$
\end{proof}

As an example application of Theorem \ref{thm:density-regularity}, we explicitly compute the degree of differentiability in the case of the quantum marginal problem for two distinguishable particles.

\begin{cor} \label{cor:density-regularity-dst}
In the setting of Corollary \ref{cor:density-formula}, assume without loss of generality that $m \le n$.  Then $p_\lambda^\mathrm{dst}$ has at least $(m^2-1)(n-1) - 2$ continuous derivatives.
\end{cor}

\begin{proof}
Recall that in the setting of Corollary \ref{cor:density-formula}, we have
$$ \tot \cong \Sigma_0 = \bigg \{ (\xi, \zeta) \in \R^{m+n} \ \bigg | \ \sum_{p=1}^m \xi_p = \sum_{q=1}^m \zeta_q = 0 \bigg \}.$$
Then $\Delta_\gog |_\tot (\xi, \zeta) = \Delta_{mn}([\xi_p + \zeta_q]_{pq=11}^{mn})$, and $\Delta_\goh(\xi, \zeta) = \Delta_m(\xi) \Delta_n(\zeta)$.

First we rewrite $\Delta_\gog |_\tot.$ Writing $jk < pq$ to indicate that the double index $jk$ precedes the double index $pq$ in lexicographic order, we have:
\begin{align*}
\Delta_{mn}([\xi_p &+ \zeta_q]_{pq=11}^{mn}) = \prod_{jk < pq} (\xi_j + \zeta_k - \xi_p - \zeta_q) \\
&= \left[ \prod_{1 \le j < p \le m} \prod_{k,q=1}^n (\xi_j + \zeta_k - \xi_p - \zeta_q) \right] \left[ \prod_{j=1}^m \prod_{1 \le k<q \le n} (\xi_j + \zeta_k - \xi_j - \zeta_q) \right] \\
&= \Delta_n(\zeta)^m \prod_{j < p}^m \prod_{k,q=1}^n (\xi_j + \zeta_k - \xi_p - \zeta_q) \\
&= \Delta_m(\xi)^n \Delta_n(\zeta)^m \prod_{j < p}^m \prod_{k<q}^n \big( \xi_j - \xi_p + (\zeta_k - \zeta_q) \big ) \big( \xi_j - \xi_p - (\zeta_k - \zeta_q) \big) \\
&= \Delta_m(\xi)^n \Delta_n(\zeta)^m \prod_{j < p}^m \prod_{k<q}^n \big((\xi_j - \xi_p)^2 - (\zeta_k - \zeta_q)^2 \big).
\end{align*}
Accordingly,
\begin{align} 
\begin{split} \label{eqn:delta-rewrite}
\Delta_{\gog / \goh}(\xi, \zeta) &= \frac{\Delta_{mn}([\xi_p + \zeta_q]_{pq=11}^{mn})}{\Delta_m(\xi) \Delta_n(\zeta)} \\
 &= \Delta_m(\xi)^{n-1} \Delta_n(\zeta)^{m-1} \prod_{j < p}^m \prod_{k<q}^n \big((\xi_j - \xi_p)^2 - (\zeta_k - \zeta_q)^2 \big).
 \end{split}
\end{align} 
From (\ref{eqn:delta-rewrite}), we observe that $ \Delta_{\gog / \goh}$ is of degree $(n^2-1)(m-1)$ in each $\xi_p$ and degree $(m^2-1)(n-1)$ in each $\zeta_q$.  Our assumption that $m \le n$ implies that $(m^2-1)(n-1) \le (n^2-1)(m-1)$.  By Lemma \ref{lem:Delta-growth}, $\ell+2$ is equal to the minimum degree of $\Delta_{\gog / \goh}$ restricted to any affine line on which $\Delta_{\gog / \goh}$ does not uniformly vanish.  Again from (\ref{eqn:delta-rewrite}), we find that this minimum degree is attained by varying a single $\zeta_q$ while keeping all other coordinates fixed, giving $\ell = (m^2-1)(n-1) -2$.
\end{proof}

To conclude this section, we state a conjecture on the polynomials $\Delta_{\gog / \goh}$ that appear in the quantum marginal problems for bosons and fermions.  Recall that for $k$ indistinguishable bosons we have $G = \SU(N)$ with $N = {{n + k - 1} \choose {k}}$, while for $k$ indistinguishable fermions we have $G = \SU(K)$ with $K = {{n} \choose {k}}$.  In both of these cases we have $H \cong \SU(n)$, with
$$ \tot \cong \bigg \{ x \in \R^n \ \bigg | \ \sum_{j=1}^n x_j = 0 \bigg \}. $$
For bosons we then have
\begin{equation}
\Delta_{\gog / \goh}(x) = \frac{1}{\Delta_n(x)}  \prod_{\substack{\alpha < \beta \\ |\alpha| = |\beta| = k}} \big( \alpha \cdot x - \beta \cdot x \big) = P_{\mathrm{bos}}(x_1, \hdots, x_n),
\end{equation}
while for fermions we have
\begin{equation}
\Delta_{\gog / \goh}(x) = \frac{1}{\Delta_n(x)}  \prod_{\substack{\alpha < \beta \\ \alpha, \beta \in \mathcal{I}}} \big( \alpha \cdot x - \beta \cdot x \big) = P_{\mathrm{fer}}(x_1, \hdots, x_n),
\end{equation}
where $\mathcal{I}$ is the set of $n$-component multi-indices of length $k$ all of whose entries are equal to 0 or 1.  Here $P_{\mathrm{bos}}$ and $P_{\mathrm{fer}}$ are homogeneous polynomials in the restriction to $\tot$ of the standard coordinates $(x_1, \hdots, x_n)$ on $\R^n$.

\begin{conj} \label{conj:bosfer-degree}
In the setting of Corollary \ref{cor:density-formula-bos} (resp.~Corollary \ref{cor:density-formula-fer}), the minimum degree $\ell + 2$ of $\Delta_{\gog / \goh}$ restricted to any affine line in $\tot$ on which it does not uniformly vanish is equal to the degree of $P_{\mathrm{bos}}$ (resp.~$P_{\mathrm{fer}}$) in any one of the individual variables $x_1, \hdots, x_n$.
\end{conj}

If the above conjecture is true, then we obtain the following explicit formulae for the degrees of differentiability of $p_\lambda^\mathrm{bos}$ and $p_\lambda^\mathrm{fer}$.

\begin{cor} \label{cor:density-regularity-bos}
For $k$ indistinguishable bosons as in Corollary \ref{cor:density-formula-bos}, if Conjecture \ref{conj:bosfer-degree} holds, then $p_\lambda^\mathrm{bos}$ has at least
$$\sum_{i = 0}^{k-1} {{n+k - i - 2} \choose {k-i}} {{n+k - i - 2} \choose {k-i-1}} - n -1$$
continuous derivatives.
\end{cor}

\begin{proof}
The degree of $P_{\mathrm{bos}}(x)$ in $x_1$ is
$$\deg_{x_1} P_{\mathrm{bos}}(x) = \deg_{x_1} \Delta_N([\alpha \cdot x]_{|\alpha| = k}) - \deg_{x_1} \Delta_n(x).$$
We have $\deg_{x_1} \Delta_n(x) = n-1$, and
\begin{align*}
\deg_{x_1} \Delta_N([\alpha \cdot x]_{|\alpha| = k}) &= \# \big \{ (\alpha, \beta) \ \big | \ |\alpha| = |\beta| = k, \ \alpha < \beta, \ \alpha_1 - \beta_1 \ne 0 \big \} \\
&= \sum_{|\alpha| = k} \# \big \{ \beta \ \big | \ |\beta| = k, \ \beta_1 > \alpha_1 \big \} \\
&= \sum_{|\alpha| = k} \sum_{j = \alpha_1 + 1}^k \# \big\{ \beta \ \big | \ |\beta| = k, \ \beta_1 = j \big\} \\
&= \sum_{|\alpha| = k} \sum_{j = \alpha_1 + 1}^k {{n+k-j-2} \choose {k-j}} \\
&= \sum_{i = 0}^{k-1} \# \big \{ \alpha \ \big | \ |\alpha| = k, \ \alpha_1 = i \big \} \sum_{j = i + 1}^k {{n+k-j-2} \choose {k-j}} \\
&= \sum_{i = 0}^{k-1} {{n+k - i - 2} \choose {k-i}} \sum_{j = i + 1}^k {{n+k-j-2} \choose {k-j}} \\
&=  \sum_{i = 0}^{k-1} {{n+k - i - 2} \choose {k-i}} {{n+k - i - 2} \choose {k-i-1}}.
\end{align*}
Assuming Conjecture \ref{conj:bosfer-degree} holds, we have $\ell = \deg_{x_1} P_{\mathrm{bos}}(x) - 2$, and the statement then follows from Theorem \ref{thm:density-regularity}.
\end{proof}

\begin{cor} \label{cor:density-regularity-fer}
For $k$ indistinguishable fermions as in Corollary \ref{cor:density-formula-fer}, if Conjecture \ref{conj:bosfer-degree} holds, then $p_\lambda^\mathrm{fer}$ has at least
$${{n-1} \choose {k}} {{n-1} \choose {k-1}} - n - 1$$
continuous derivatives.
\end{cor}

\begin{proof}
The degree of $P_{\mathrm{fer}}(x)$ in $x_1$ is
\begin{align*}
\deg_{x_1} P_{\mathrm{fer}}(x) &= \deg_{x_1} \Delta_K([\alpha \cdot x]_{\alpha \in \mathcal{I}}) - \deg_{x_1} \Delta_n(x) \\
&= \# \big \{ (\alpha, \beta) \in \mathcal{I}^2 \ | \ \alpha < \beta, \ \alpha_1 - \beta_1 \ne 0 \big \} - (n-1) \\
&= \# \big \{ (\alpha, \beta) \in \mathcal{I}^2 \ | \ \alpha_1 = 0, \ \beta_1 = 1 \big \} - (n-1) \\
&= {{n-1} \choose {k}} {{n-1} \choose {k-1}} - n + 1.
\end{align*}
Again assuming Conjecture \ref{conj:bosfer-degree} holds, we have $\ell = \deg_{x_1} P_{\mathrm{fer}}(x) - 2$, and the statement then follows from Theorem \ref{thm:density-regularity}.
\end{proof}

\section{Restriction multiplicities and the semiclassical limit} \label{sec:multiplicities}

We now turn our attention to a combinatorial problem in representation theory that is closely related to projections of orbital measures, namely the problem of computing restriction multiplicities, which we define below.  It is well known in representation theory and symplectic geometry that projections of orbital measures provide a kind of asymptotic approximation of restriction multiplicities.  Here we will study an exact, rather than asymptotic, relation between these two objects, from which we will deduce the asymptotic approximation as a corollary.  Using this exact relation, we prove a new integral formula for the restriction multiplicities and show that the measure $\psi_* \phi_* \beta_\lambda$ studied in the preceding sections arises as a limit of point processes that generalize the Littlewood--Richardson and $\U(n)$ branching processes introduced by Biane \cite{Biane95}.  The discussion below also generalizes several results on the randomized Horn's problem that were shown in \cite{CZ1, CMZ, McS-splines}.  In the process, we review a number of ideas that have previously appeared in the literature, starting with the seminal work of Heckman \cite{Heck} and culminating in more recent papers of Vergne and several collaborators \cite{DPV, DV, PV}; see also \cite[\textsection6]{CDKW} and references therein.

The finite-dimensional irreducible representations of $G$ and $H$ are parametrized by their respective cones of dominant integral weights $P_G^+ \subset \tilde \tot_+$ and $P_H^+ \subset \tot_+$.  Given $\lambda \in P_G^+$, we can decompose the corresponding irreducible representation $V^G_\lambda$ of $G$ as a direct sum of irreducible representations of $H$:
$$V^G_\lambda = \bigoplus_{\mu \in P_H^+} (V^H_\mu)^{m^\lambda_\mu}.$$
That is, the representation $V^H_\mu$ occurs in $V^G_\lambda$ with multiplicity $m^\lambda_\mu$.  This decomposition can be equivalently stated as an identity of characters:
$$\chi^G_\lambda \Big |_H = \sum_{\mu \in P_H^+} m^\lambda_\mu \, \chi^H_\mu,$$
where $\chi^G_\lambda |_H$ is the restriction of the character of $V^G_\lambda$ to $H \subset G$.  Accordingly, the multiplicities $m^\lambda_\mu$ are called {\it restriction multiplicities}.

Many quantities of combinatorial interest are special cases of restriction multiplicities.  For example, when $H$ is a maximal torus of $G$, the restriction multiplicities are the weight multiplicities of irreducible representations of $G$.  For $G = \SU(n)$ these coincide with the Kostka numbers, which count the number of semistandard Young tableaux of a given shape and weight.  When $G = K \times K$ for some compact connected Lie group $K$, and $H \cong K$ is the diagonal subgroup, the restriction multiplicities describe the decomposition of tensor products of irreducible representations of $K$.  For $K = \SU(n)$ these coincide with the structure constants of Schur polynomials, known as Littlewood--Richardson coefficients.  For $G = \U(n^2)$ and $H = \U(n)^{\otimes 2}$, as in the two-particle quantum marginal problem studied in Section \ref{sec:formula-dist} with $m = n$, the restriction multiplicities are known to coincide with the tensor product multiplicities for the symmetric group $S_n$, known as Kronecker coefficients; see \cite[\textsection5]{Klyachko-marginals} and \cite[\textsection7]{CDKW}.  Algorithms for computing restriction multiplicities are known, but in general the problem is computationally intractable; see e.g.~\cite{HN}.

A celebrated idea in symplectic geometry due to Heckman \cite{Heck} states that projections of orbital measures describe the asymptotic behavior of restriction multiplicities.  To state this relationship precisely, we define the function
\begin{equation} \label{eqn:J-def}
\mathcal{J}_\lambda(x) = \frac{\Delta_\goh(\rho_\goh)}{\Delta_\gog(\rho_\gog)} \frac{\Delta_\gog(\lambda)}{\Delta_\goh(x)} p_{\lambda}(x), \qquad x \in \tot_+,
\end{equation}
for $\lambda \in \tilde \tot$ with $\Delta_\gog(\lambda) \ne 0$, where $p_{\lambda}$ is the probability density described in Section \ref{sec:projections-general}.  Throughout this section we assume for the sake of simplicity that sufficient conditions are met for the density $p_\lambda$ to exist, but if this is not the case then we can define $\mathcal{J}_\lambda$ as a function on a lower-dimensional affine subspace following the discussion in Remark \ref{rem:unfaithful}, and analogues of the statements below still hold.  If $\Delta_\gog(\lambda) = 0$ then we set $\mathcal{J}_\lambda \equiv 0$.  We extend $\mathcal{J}_\lambda$ antisymmetrically to the other Weyl chambers in $\tot$ by setting $\mathcal{J}_\lambda(w(x)) = \epsilon(w) \mathcal{J}_\lambda(x)$ for $x \in \tot_+$.

The function $\mathcal{J}_\lambda$ thus defined is essentially the same as the asymptotic multiplicity function defined in \cite{Heck} and generalizes the volume function studied in \cite{CZ1, CMZ, McS-splines}.  It is important to note that, since $p_\lambda$ is the density of a measure, in general the definition (\ref{eqn:J-def}) only determines $\mathcal{J}_\lambda$ almost everywhere with respect to Lebesgue measure.  However, in most cases of interest $p_\lambda$ is actually continuous, and we then can take (\ref{eqn:J-def}) to be a pointwise equality.

Let $Q_G \subset \tilde \tot$ be the lattice spanned by $\Phi^+_\gog$, and write $\Pi$ for the orthogonal projection from $\tilde \tot$ onto $\tot$.  It is well known that $m^\lambda_\mu = 0$ unless $\mu \in \Pi(Q_G + \lambda) \cap P_H^+$; this follows for example from the multiplicity formula \cite[eqn.~(3.5)]{Heck}. Set $d = |\Phi_\gog^+| - |\Phi_\goh^+| - r$, where again we write $r = \dim \tot$ for the rank of $\goh$, and suppose that $\mathcal{J}_\lambda$ is continuous. Then, from \cite[Theorem 3.9]{Heck}, we have the following asymptotic relation between $\mathcal{J}_\lambda$ and the multiplicities $m^\lambda_\mu$:
\begin{equation} \label{eqn:mult-semiclassical}
\mathcal{J}_\lambda(\mu) = \lim_{n \to \infty} n^{-d} \, m^{n \lambda}_ {n \mu}, \qquad \lambda \in P_G^+, \quad \mu \in \Pi(Q_G + \lambda) \cap P_H^+.
\end{equation}
The assumption of continuity is not hard to remove, but doing so requires some care in defining the values of $\mathcal{J}_\lambda$ at points where it is discontinuous.  In the remainder of this section we will not need this assumption.

In the theory of geometric quantization as developed by Guillemin--Sternberg \cite{GS}, the multiplicities $m^\lambda_\mu$ are interpreted as physical quantities associated to a quantum mechanical system,\footnote{This quantum mechanical interpretation of the restriction multiplicities is distinct from, and not directly related to, the quantum marginal problems studied in Sections \ref{sec:formula-dist}--\ref{sec:formula-fer}.  Indeed, in the sense of geometric quantization, these quantum marginal problems are {\it classical} versions of corresponding multiplicity problems. In the case of Section \ref{sec:formula-dist}, when $m=n$ the relevant multiplicities are Kronecker coefficients.} and the limit in (\ref{eqn:mult-semiclassical}) corresponds to a semiclassical approximation in which $\hbar$ tends to 0.  Accordingly, one says that $\mathcal{J}_\lambda$ is the {\it semiclassical limit} of the multiplicities.  Here we will develop an alternate perspective on (\ref{eqn:mult-semiclassical}) based on a convolution identity due, in various forms, to Duflo--Vergne \cite{DV}, Paradan--Vergne \cite{PV} and De Concini--Procesi--Vergne \cite{DPV}.  The idea is that for $\lambda \in P_G^+$, the function $\mathcal{J}_{\lambda + \rho_\gog}$ can be expressed as the convolution of a certain spline function with a discretely supported measure that encodes the multiplicities $m^\lambda_\mu$.

The spline function in question is the {\it box spline} $\mathcal{B}_{\gog / \goh}$, which we now define.  We will need some further notation.  Define the periodic functions
$$\widehat{\Delta}_\gog(x) = \prod_{\alpha \in \Phi_\gog^+} (e^{i \langle \alpha, x \rangle/2} - e^{- i\langle \alpha,x \rangle/2}), \qquad x \in \tilde \tot,$$
$$\widehat{\Delta}_\goh(x) = \prod_{\alpha \in \Phi_\goh^+} (e^{i \langle \alpha, x \rangle/2} - e^{- i\langle \alpha,x \rangle/2}), \qquad x \in \tot,$$
and $\widehat{\Delta}_{\gog/\goh}(x) =\widehat{\Delta}_{\gog}(x) / \widehat{\Delta}_{\goh}(x)$, $x \in \tot$.  One way to define $\mathcal{B}_{\gog / \goh}$ is as the inverse Fourier transform\footnote{Here, as above, in cases where the integral is not absolutely convergent we interpret it as a principal value.  This allows us to assign pointwise values to $\mathcal{B}_{\gog / \goh}$ even when it is not continuous.} of $\widehat{\Delta}_{\gog/\goh}(x) / \Delta_{\gog/\goh}(ix)$:
\begin{equation} \label{eqn:B-def}
\mathcal{B}_{\gog / \goh}(x) = \frac{1}{(2\pi)^{r}} \int_\tot \frac{\widehat{\Delta}_{\gog/\goh}(\xi)}{\Delta_{\gog / \goh}(i\xi)} \, e^{-i \langle x, \xi \rangle} \, d\xi.
\end{equation}
Equivalently, the box spline is the density of a probability measure -- which, in a slight abuse of notation, we also denote $\mathcal{B}_{\gog / \goh}$ -- defined as follows in terms of its pairing with continuous functions on $\tot$:
\begin{equation} \label{eqn:B-def2}
\big( \mathcal{B}_{\gog / \goh}, f \big) = \int_{-1/2}^{1/2} \cdots \int_{-1/2}^{1/2} f \Big( \sum_{\alpha \in \Phi^+_{\gog / \goh}} t_\alpha \alpha \Big) \prod_{\alpha \in \Phi^+_{\gog / \goh}} dt_\alpha, \qquad f \in C^0(\tot).
\end{equation}
In words, $\mathcal{B}_{\gog / \goh}$ is a convolution of uniform probability measures on the line segments $[ - \alpha / 2, \ \alpha / 2]$ for $\alpha \in \Phi^+_{\gog / \goh}$.  It is apparent that the resulting density is a piecewise polynomial function supported on the zonotope
$$ \bigg\{ \sum_{\alpha \in \Phi^+_{\gog / \goh}} t_\alpha \alpha \ \ \bigg | \ \  t_\alpha \in \big[ -1/2, \ 1/2 \big] \bigg \} \subset \tot. $$
The first definition (\ref{eqn:B-def}) can be obtained directly from (\ref{eqn:B-def2}) by taking $f(x) = e^{i \langle x, \xi \rangle}$.

For $\lambda \in P_G^+$, the following proposition expresses $\mathcal{J}_{\lambda + \rho_\gog}$ as a convolution of $\mathcal{B}_{\gog / \goh}$ and a finitely supported, $W_\goh$-skew measure that encodes the multiplicities $m^\lambda_\mu$.  The statement below is a special case of a much more general result in index theory shown in \cite[Proposition 5.14]{DPV}; here we give a direct proof adapted to the case at hand, based on an argument sketched in \cite[Lemma 5.2]{PV}.

\begin{prop} \label{prop:spline-conv}
For $\lambda \in P_G^+$, define a measure $M_\lambda$ on $\tot$ by
\begin{equation} \label{eqn:Mlambda-def}
M_\lambda = \sum_{\mu \in P_H^+} m^\lambda_\mu \sum_{w \in W_\goh} \epsilon(w) \, \delta_{w(\mu + \rho_\goh)}.
\end{equation}
Then we have the convolution identity
\begin{equation} \label{eqn:spline-conv}
\mathcal{J}_{\lambda + \rho_\gog}(x) = \big( \mathcal{B}_{\gog / \goh} * M_\lambda \big)(x).
\end{equation}
\end{prop}

Explicitly, (\ref{eqn:spline-conv}) means that for $\lambda \in P_G^+$,
\begin{equation} \label{eqn:spline-conv2}
\mathcal{J}_{\lambda + \rho_\gog}(x) = \sum_{\mu \in P_H^+} m^\lambda_\mu \sum_{w \in W_\goh} \epsilon(w) \, \mathcal{B}_{\gog / \goh} \big(x - w(\mu + \rho_\goh) \big)
\end{equation}
almost everywhere with respect to Lebesgue measure on $\tot$.  But in fact, we can always take (\ref{eqn:spline-conv2}) to be a pointwise equality if we like: when $\mathcal{J}_{\lambda + \rho_\gog}$ (or $\mathcal{B}_{\gog / \goh}$) is continuous, (\ref{eqn:spline-conv2}) holds pointwise, and if $\mathcal{J}_{\lambda + \rho_\gog}$ is not continuous, we can {\it define} its pointwise values via (\ref{eqn:spline-conv2}) and (\ref{eqn:B-def}).

\begin{proof}
We first recall the Kirillov character formula for compact Lie groups, which expresses the character $\chi^G_\lambda$ in terms of a Harish-Chandra integral of the form (\ref{eqn:hc-integral}):
\begin{equation} \label{eqn:KCF}
\frac{\widehat{\Delta}_\gog(x)}{\Delta_\gog(ix)} \chi^G_\lambda(e^x) = \frac{\Delta_\gog(\lambda + \rho_\gog)}{\Delta_\gog(\rho_\gog)} \mathcal{H}_G(\lambda + \rho_\gog, ix ), \qquad x \in \tilde \tot,
\end{equation}
where the exponential notation $e^x$ merely indicates that we regard the character as a function on the group rather than the algebra.  See \cite[ch.~5]{Kirillov-lectures} for a detailed discussion of the Kirillov character formula.

Restricting (\ref{eqn:KCF}) to $x \in \tot$ and decomposing $\chi^G_\lambda$ into irreducible characters of $H$, we find
\begin{equation} \label{eqn:mult1}
\frac{\widehat{\Delta}_\gog(x)}{\Delta_\gog(ix)} \sum_{\mu \in P_H^+} m^\lambda_\mu \, \chi^H_\mu(e^x) = \frac{\Delta_\gog(\lambda + \rho_\gog)}{\Delta_\gog(\rho_\gog)} \mathcal{H}_G(\lambda + \rho_\gog, ix ), \qquad x \in \tot.
\end{equation}
Applying (\ref{eqn:KCF}) again on the left-hand side above, and then applying the Harish-Chandra formula (\ref{eqn:hc-integral}), we obtain
\begin{align*}
\frac{\widehat{\Delta}_\gog(x)}{\Delta_\gog(ix)} \sum_{\mu \in P_H^+} m^\lambda_\mu \, \chi^H_\mu(e^x) &= \frac{\widehat{\Delta}_{\gog/\goh}(x)}{\Delta_{\gog/\goh}(ix)} \sum_{\mu \in P_H^+} m^\lambda_\mu \, \frac{\Delta_\goh(\mu+\rho_\goh)}{\Delta_\goh(\rho_\goh)} \mathcal{H}_H(\mu + \rho_\goh, ix ) \\
&= \frac{\widehat{\Delta}_{\gog/\goh}(x)}{\Delta_{\gog/\goh}(ix)} \sum_{\mu \in P_H^+}  \frac{m^\lambda_\mu}{\Delta_\goh(ix)} \sum_{w \in W_\goh} \epsilon(w) e^{i \langle w(\mu + \rho_\goh), x \rangle},
\end{align*}
so that after multiplying through by $\Delta_\goh(ix)$, (\ref{eqn:mult1}) becomes
$$ \frac{\widehat{\Delta}_{\gog/\goh}(x)}{\Delta_{\gog/\goh}(ix)} \sum_{\mu \in P_H^+} m^\lambda_\mu \sum_{w \in W_\goh} \epsilon(w) e^{i \langle w(\mu + \rho_\goh), x \rangle} = \frac{\Delta_\gog(\lambda + \rho_\gog)}{\Delta_\gog(\rho_\gog)} \Delta_\goh(ix) \mathcal{H}_G(\lambda + \rho_\gog, ix ). $$
Finally, taking the inverse Fourier transform of both sides over $\tot$ and applying Corollary \ref{cor:pushfwd-FT} and the definition (\ref{eqn:J-def}), the equation above becomes (\ref{eqn:spline-conv}).
\end{proof}

\begin{remark} \label{rem:spline-regularity}
It is obvious from Proposition \ref{prop:spline-conv} that $\mathcal{J}_{\lambda + \rho_\gog}$ has at least as many continuous derivatives as $\mathcal{B}_{\gog / \goh}$ whenever $\lambda \in P_G^+$.  In fact, it is easy to show\footnote{One can use an approximation argument as in \cite[Corollary 4.1]{McS-splines}.} that $\mathcal{J}_\lambda$ has at least this many continuous derivatives for all $\lambda \in \tilde \tot$. By a standard fact about spline functions (see e.g.~\cite[\textsection1.5]{PBP}), if all subsets of $\Phi^+_{\gog/\goh}$ obtained by deleting $\ell+1$ vectors span $\tot$, then $\mathcal{B}_{\gog / \goh} \in C^\ell(\tot)$.  We thus obtain an alternate proof of Theorem \ref{thm:density-regularity}.
\end{remark}

Proposition \ref{prop:spline-conv} yields a distributional version of the semiclassical limit (\ref{eqn:mult-semiclassical}).  For $s > 0$, we define the rescaling operator $\mathcal{R}_s$, which acts on measures or functions on $\tot$ as follows.  If $\mathcal{M}$ is a (possibly signed) measure on $\tot$, then $\mathcal{R}_s \mathcal{M}(E) = \mathcal{M}(sE)$ for Borel sets $E \subset \tot$, where $sE$ is the dilation of $E$ by a factor of $s$.  If $f$ is a function on $\tot$, then $\mathcal{R}_s f(x) = s^{r} f(sx)$, where $r = \dim \tot$.  With this definition, if $\mathcal{M}$ has density $f$ with respect to Lebesgue measure, then $\mathcal{R}_s \mathcal{M}$ has density $\mathcal{R}_s f$.  Write $dx$ for Lebesgue measure on $\tot$ and $\implies$ for weak convergence of finite signed measures.

\begin{cor} \label{cor:sc-limit-dist}
For $\lambda \in P_G^+$, as $n \to \infty$,
\begin{equation} \label{eqn:sc-limit-dist}
n^{-|\Phi^+_{\gog / \goh}|} \mathcal{R}_n M_{n \lambda} \implies \mathcal{J}_\lambda(x) \, dx.
\end{equation}
\end{cor}

To prove Corollary \ref{cor:sc-limit-dist}, we will use a homogeneity property of the function $\mathcal{J}_\lambda$.  The following lemma is essentially equivalent to \cite[eqn.~(3.24)]{Heck}, but here we give a concise alternate proof.

\begin{lem} \label{lem:J-homog}
As a function of the pair $(\lambda, x) \in \tilde \tot \times \tot$, $\mathcal{J}_\lambda(x)$ is homogeneous of degree $d = |\Phi_\gog^+| - |\Phi_\goh^+| - r$.  That is, for $s > 0$, $\mathcal{J}_{s \lambda}(sx) = s^d \mathcal{J}_\lambda(x).$
\end{lem}

\begin{proof}
We first note that if $\Delta_\gog(\lambda) = 0$ then $\mathcal{J}_{s \lambda} \equiv 0$ for all $s > 0$ and there is nothing to prove.  Thus we assume $\Delta_\gog(\lambda) \not = 0$.  For the sake of simplicity in the calculations that follow we also assume as in Theorem \ref{thm:integral-formula-gen} that $\Delta_\gog |_\tot$ is not identically zero; if this assumption does not hold then we can apply L'H\^opital's rule as in Remark \ref{rem:unfaithful} and proceed by the same argument as below.

From the definition (\ref{eqn:J-def}) of $\mathcal{J}_\lambda$ in terms of $p_\lambda$ and from the integral formula (\ref{eqn:integral-formula-gen}) for $p_\lambda$, we have
\begin{equation} \label{eqn:J-integral}
\mathcal{J}_\lambda(x) = \frac{(-i)^{|\Phi_\gog^+| - |\Phi_\goh^+|}}{(2\pi)^{r}} \int_\tot \frac{\Delta_\goh(\xi)}{\Delta_\gog(\xi)}  \sum_{w \in W_\gog} \epsilon(w) \, e^{i \langle \xi, w(\lambda) - x \rangle} \, d\xi,
\end{equation}
so that
\begin{align*}
\mathcal{J}_{s \lambda}(sx) &= \frac{(-i)^{|\Phi_\gog^+| - |\Phi_\goh^+|}}{(2\pi)^{r}} \int_\tot \frac{\Delta_\goh(\xi)}{\Delta_\gog(\xi)}  \sum_{w \in W_\gog} \epsilon(w) \, e^{i \langle \xi, w(s \lambda) - sx \rangle} \, d\xi \\
&= \frac{(-i)^{|\Phi_\gog^+| - |\Phi_\goh^+|}}{(2\pi)^{r}} \int_\tot \frac{\Delta_\goh(\xi)}{\Delta_\gog(\xi)}  \sum_{w \in W_\gog} \epsilon(w) \, e^{i \langle s \xi, w(\lambda) - x \rangle} \, d\xi \\
&= \frac{(-i)^{|\Phi_\gog^+| - |\Phi_\goh^+|}}{(2\pi)^{r}} \int_\tot \frac{\Delta_\goh(s^{-1} \xi)}{\Delta_\gog(s^{-1} \xi)}  \sum_{w \in W_\gog} \epsilon(w) \, e^{i \langle \xi, w(\lambda) - x \rangle} \, d(s\xi).
\end{align*}
The statement then follows from the homogeneity properties
$$\Delta_\goh(s^{-1} \xi) = s^{-|\Phi_\goh^+|} \Delta_\goh(\xi), \quad \Delta_\gog(s^{-1} \xi) = s^{-|\Phi_\gog^+|} \Delta_\gog(\xi), \quad d(s \xi) = s^{-r} d\xi.$$
\end{proof}

\begin{proof}[Proof of Corollary \ref{cor:sc-limit-dist}]
From Proposition \ref{prop:spline-conv}, we have $$\mathcal{J}_{n\lambda + \rho_\gog} = \mathcal{B}_{\gog / \goh} * M_{n \lambda}.$$  Applying $\mathcal{R}_n$ on both sides and using Lemma \ref{lem:J-homog}, we find
  \begin{equation} \label{eqn:conv-intermediate}
      \mathcal{R}_n \mathcal{J}_{n\lambda + \rho_\gog} = \mathcal{R}_n \big( \mathcal{B}_{\gog / \goh} * M_{n \lambda} \big) = \mathcal{R}_n \mathcal{B}_{\gog / \goh} * \mathcal{R}_n M_{n \lambda} = n^{d+r} \mathcal{J}_{\lambda + \rho_\gog / n},
  \end{equation}
where $\mathcal{R}_n$ acts on $\mathcal{B}_{\gog / \goh}$ as a measure and on $\mathcal{J}_{n\lambda + \rho_\gog}$ as a function.  Thus we have
\begin{align*}
\lim_{n \to \infty} n^{-(d+ r)} \big( \mathcal{R}_n \mathcal{B}_{\gog / \goh} * \mathcal{R}_n M_{n \lambda} \big) &= \lim_{n \to \infty} \big[ \mathcal{R}_n \mathcal{B}_{\gog / \goh} * \big( n^{-(d + r)} \mathcal{R}_n M_{n \lambda} \big) \big] \\
&= \lim_{n \to \infty} \mathcal{J}_{\lambda + \rho_\gog / n} = \mathcal{J}_\lambda
\end{align*}
almost everywhere with respect to Lebesgue measure on $\tot$, which is a strictly stronger statement than the weak limit
\begin{equation} \label{eqn:wlim-approx}
\mathcal{R}_n \mathcal{B}_{\gog / \goh} * \big( n^{-(d + r)} \mathcal{R}_n M_{n \lambda} \big) \implies \mathcal{J}_\lambda(x) \, dx.
 \end{equation}
Since $\mathcal{B}_{\gog / \goh}$ is a probability measure, the sequence $\mathcal{R}_n \mathcal{B}_{\gog / \goh}$ is an approximation to the identity as $n \to \infty$.  Accordingly, the sequence of measures $n^{-(d + r)} \mathcal{R}_n M_{n \lambda}$ must have the same weak limit as the left-hand side of (\ref{eqn:wlim-approx}), which is the desired result, since $d + r = |\Phi^+_\gog| - |\Phi^+_\goh| = |\Phi^+_{\gog / \goh}|.$
\end{proof}

We can also interpret the pushforward measure $\psi_* \phi_* \beta_\lambda$ itself as a semiclassical limit.  For $\lambda \in P^+_G,$ define the probability measure
\begin{equation} \label{eqn:Xilambda-def}
\Xi_\lambda = \frac{1}{\dim V^G_\lambda} \sum_{\mu \in P^+_H} m^\lambda_\mu \, \dim V^H_\mu \, \delta_\mu.
\end{equation}
The measure $\Xi_\lambda$ is a generalization of the Littlewood--Richardson process studied by Biane \cite{Biane95} and Collins--Novak--\'Sniady \cite{CNS}, which corresponds to the case $G = \U(n)^2$ with $H \cong \U(n)$ the diagonal subgroup.  The case $G = \U(n)$, $H = \U(m)$, $m < n$ was also studied in \cite{Biane95} and is sometimes called the $\U(n)$ branching process.  When $G = \U(n)$, the following theorem can be deduced from the result shown in \cite[\textsection2.4]{Biane95}.  Here we show a general result for compact connected Lie groups, using very different methods based on box spline convolution.

\begin{thm} \label{thm:sc-limit-prob}
For $\lambda \in P^+_G,$ we have the weak limit
\begin{equation} \label{eqn:sc-limit-prob}
\mathcal{R}_n \Xi_{n \lambda} \implies \psi_* \phi_* \beta_\lambda
\end{equation}
as $n \to \infty$.
\end{thm}
Note that Theorem \ref{thm:sc-limit-prob} applies even when $\Delta_\gog(\lambda) = 0$.

\begin{proof}
If $v \in \tot$ and $\mathcal{M}$ is a signed measure on $\tot$, write $\mathcal{T}_v \mathcal{M}$ for the translation of $\mathcal{M}$ by $v$.  That is, $\mathcal{T}_v \mathcal{M}(E) = \mathcal{M}(E - v)$ for Borel sets $E \subset \tot$.  If $\mathcal{M}$ is supported on $\tot_+$, write $\mathcal{S} \mathcal{M}$ for its $W_\goh$-skew extension to $\tot$, that is, $\mathcal{S} \mathcal{M}(w(E)) = \epsilon(w) \mathcal{M}(E)$ for Borel sets $E \subset \tot_+$.  We have the following relations between $\mathcal{S}$, $\mathcal{T}_v$, and the rescaling operator $\mathcal{R}_n$:
\begin{equation} \label{eqn:rescaling-relns}
\mathcal{R}_n \mathcal{S} = \mathcal{S} \mathcal{R}_n, \qquad \mathcal{R}_n \mathcal{T}_v = \mathcal{T}_{v/n} \mathcal{R}_n.
\end{equation}

Recall the Weyl dimension formula:
\begin{equation} \label{eqn:weyl-dim}
\dim V^G_\lambda = \frac{\Delta_\gog(\lambda + \rho_\gog)}{\Delta_\gog(\rho_\gog)}.
\end{equation}
Using (\ref{eqn:weyl-dim}), we find that the measure $M_\lambda$ defined in (\ref{eqn:Mlambda-def}) and the measure $\Xi_\lambda$ defined in (\ref{eqn:Xilambda-def}) are related by
\begin{equation} \label{eqn:M-Xi}
M_\lambda(dx) = \frac{\Delta_\goh(\rho_\goh)}{\Delta_\gog(\rho_\gog)} \frac{\Delta_\gog(\lambda + \rho_\gog)}{\Delta_\goh(x)} \, \mathcal{S} \mathcal{T}_{\rho_\goh} \Xi_\lambda (dx).
\end{equation}
Note that $\Delta_\gog(\lambda + \rho_\gog / n) \ne 0$, so from the definition (\ref{eqn:J-def}) of $\mathcal{J}_\lambda$ we have
\begin{equation} \label{eqn:J-pshfwd}
\mathcal{S} \psi_* \phi_* \beta_{\lambda + \rho_\gog / n} (dx) = \frac{\Delta_\gog(\rho_\gog)}{\Delta_\goh(\rho_\goh)} \frac{\Delta_\goh(x)}{\Delta_\gog(\lambda + \rho_\gog / n)} \mathcal{J}_{\lambda + \rho_\gog / n}(x) \, dx.
\end{equation}
Using (\ref{eqn:M-Xi}) and (\ref{eqn:J-pshfwd}), we can rewrite the final equality of (\ref{eqn:conv-intermediate}) as
\begin{align*}
\mathcal{S} \psi_* \phi_* \beta_{\lambda + \rho_\gog / n} &= \mathcal{R}_n \mathcal{B}_{\gog / \goh} * \mathcal{R}_n \mathcal{S} \mathcal{T}_{\rho_\goh} \Xi_{n \lambda} \\
&= \mathcal{R}_n \mathcal{B}_{\gog / \goh} * \mathcal{S} \mathcal{T}_{\rho_\goh / n} \mathcal{R}_n \Xi_{n \lambda},
\end{align*}
where in the second equality above we have used the relations (\ref{eqn:rescaling-relns}).  Since $\mathcal{R}_n \mathcal{B}_{\gog / \goh}$ is an approximation to the identity as $n \to \infty$, the sequence of measures $\mathcal{S} \mathcal{T}_{\rho_\goh / n} \mathcal{R}_n \Xi_{n \lambda}$ must have the same weak limit as the sequence $\mathcal{S} \psi_* \phi_* \beta_{\lambda + \rho_\gog / n}$, namely $\mathcal{S} \psi_* \phi_* \beta_\lambda$.  Restricting to $\tot_+$, we find that $\mathcal{T}_{\rho_\goh / n} \mathcal{R}_n \Xi_{n \lambda} \Rightarrow \psi_* \phi_* \beta_\lambda$ and therefore also $\mathcal{R}_n \Xi_{n \lambda} \Rightarrow \psi_* \phi_* \beta_\lambda$ as desired.
\end{proof}

Finally, we briefly discuss the question of inverting the semiclassical limit, that is, recovering the multiplicities $m^\lambda_\mu$ from the function $\mathcal{J}_\lambda$.  Such a computation amounts to inverting the convolution relation (\ref{eqn:spline-conv}).  Somewhat surprisingly, this can be accomplished by several different methods.  The papers \cite{DV, PV, DPV} use box spline deconvolution results of Dahmen--Micchelli \cite{DM2, DM} to develop techniques for computing the multiplicities from $\mathcal{J}_\lambda$ by applying a differential operator.  In \cite{McS-splines}, the second author presented multiple complementary approaches for recovering Littlewood--Richardson coefficients, Kostka numbers, and other tensor product and weight multiplicities from their semiclassical limits.  These techniques include integral formulae, a deconvolution algorithm, and a representation of $\mathcal{J}_\lambda$ as a solution of a finite difference equation.  All of these approaches can be generalized to cover more general restriction multiplicities; for instance, the convolution identity (\ref{eqn:spline-conv}) can be inverted computationally via the algorithm of \cite[Lemma 5.2]{McS-splines}.  As a further example, here we show an integral formula expressing $m^\lambda_\mu$ in terms of finitely many values of $\mathcal{J}_{\lambda + \rho_\gog}$.

Before stating the formula, we introduce some final notation.  Recall that $\Pi$ is the orthogonal projection from $\tilde \tot$ to $\tot$, and $Q_G \subset \tilde \tot$ is the lattice spanned by $\Phi^+_\gog$.  Then $\Pi(Q_G)$ is a lattice in $\tot$.  Write $|\Pi(Q_G)|$ for the volume of a fundamental domain of $\Pi(Q_G)$, and $\Pi(Q_G)^*$ for the dual lattice.  The character group of $\Pi(Q_G)$ is the torus $T \cong \tot / 2\pi \Pi(Q_G)^*$, which we may identify with any fundamental domain of the lattice $2\pi \Pi(Q_G)^* \subset \tot$.  Define
\begin{equation} \label{eqn:Rpoly-def}
R(x) = \sum_{\nu \in \Pi(Q_G)} \mathcal{B}_{\gog / \goh}(\nu) \, e^{i \langle \nu, x \rangle}, \qquad x \in \tot.
\end{equation}
The function $R(x)$ generalizes the $R$-polynomial of a semisimple Lie algebra, which arises in the study of tensor product multiplicities.  The $R$-polynomial was introduced in \cite{CZ1} and further studied in \cite{Coq2, CMZ, ER, McS-splines}.  Although the sum in (\ref{eqn:Rpoly-def}) runs over the entire lattice $\Pi(Q_G)$, in fact only finitely many terms are nonzero, since $\mathcal{B}_{\gog / \goh}$ is compactly supported.  Thus $R(x)$ is a $\Pi(Q_G)$-periodic trigonometric polynomial on $\tot$.  The following integral formula generalizes \cite[Theorem 5.5]{McS-splines}.

\begin{thm} \label{thm:mult-integral}
For $\lambda \in P^+_G$ and $\mu \in P^+_H$,
\begin{equation} \label{eqn:mult-integral}
m^\lambda_\mu = \frac{|\Pi(Q_G)|}{(2 \pi)^{r}} \int_{T} \frac{1}{R(x)} \sum_{\substack{\nu \in \\ \Pi(\lambda + Q_G)}} \mathcal{J}_{\lambda + \rho_\gog} (\nu + \rho_\goh) \, e^{i \langle \nu - \mu, x \rangle} \, dx.
\end{equation}
\end{thm}
Similarly to (\ref{eqn:Rpoly-def}), the sum in (\ref{eqn:mult-integral}) involves only finitely many nonzero terms due to the fact that $\mathcal{J}_{\lambda + \rho_\gog}$ is compactly supported.
\begin{proof}
At points of the form $x = \nu + \rho_\goh$ with $\nu \in \Pi(\lambda + Q_G)$, the convolution identity (\ref{eqn:spline-conv2}) reads:
\begin{align} \begin{split} \label{eqn:J-nu-rho}
\mathcal{J}_{\lambda + \rho_\gog}(\nu + \rho_\goh) &= \sum_{\mu \in P_H^+} m^\lambda_\mu \sum_{w \in W_\goh} \epsilon(w) \, \mathcal{B}_{\gog / \goh} \big(\nu + \rho_\goh - w(\mu + \rho_\goh) \big) \\
 &= \sum_{\mu \in P_H^+ \cap \Pi(\lambda + Q_G)} m^\lambda_\mu \sum_{w \in W_\goh} \epsilon(w) \, \mathcal{B}_{\gog / \goh} \big(\nu + \rho_\goh - w(\mu + \rho_\goh) \big),
\end{split} \end{align}
where the second equality follows from the fact that $m^\lambda_\mu = 0$ unless $\mu \in \Pi(\lambda + Q_G)$.  Working in the fundamental weight basis and writing $w$ as a product of reflections, a direct calculation reveals that if $\mu, \nu \in  \Pi(\lambda + Q_G)$ then $\nu + \rho_\goh - w(\mu + \rho_\goh) \in \Pi(Q_G)$.  Accordingly, the final line in the display above only involves values of $\mathcal{B}_{\gog / \goh}$ at points of $\Pi(Q_G)$.  Multiplying the first and last expressions in (\ref{eqn:J-nu-rho}) by a Dirac mass at $\nu + \rho_\goh$ and summing over all $\nu \in \Pi(\lambda + Q_G)$, we obtain the following identity of signed measures:
\begin{equation} \label{eqn:disc-conv}
\sum_{\nu \in \Pi(\lambda + Q_G)} \mathcal{J}_{\lambda + \rho_\gog}(\nu + \rho_\goh) \, \delta_{\nu + \rho_\goh} = \bigg( \sum_{\nu \in \Pi(Q_G)} \mathcal{B}_{\gog / \goh}(\nu) \, \delta_\nu \bigg) * M_\lambda.
\end{equation}
Note that the measure expressed on either side of (\ref{eqn:disc-conv}) is supported on finitely many points of $\tot$.  Taking the Fourier transform of both sides and dividing through by $R(x)$, we find
\begin{equation} \label{eqn:disc-conv-FT}
\sum_{\mu \in P_H^+} m^\lambda_\mu \sum_{w \in W_\goh} \epsilon(w) \, e^{i \langle w(\mu + \rho_\goh), x \rangle} = \frac{1}{R(x)} \sum_{\nu \in \Pi(\lambda + Q_G)} \mathcal{J}_{\lambda + \rho_\gog}(\nu + \rho_\goh) \, e^{i \langle \nu + \rho_\goh, x \rangle}
\end{equation}
at all points $x \in T$ with $R(x) \ne 0$.  Since $R(x)$ is a nonzero trigonometric polynomial, it can vanish at most on a set of measure zero, so that (\ref{eqn:disc-conv-FT}) holds almost everywhere with respect to Haar measure on $T$.  From inspecting the left-hand side, it is clear that the function expressed in (\ref{eqn:disc-conv-FT}) is bounded and therefore integrable on $T$.  Thus we can extract the coefficient $m^\lambda_\mu$ of $e^{i \langle \mu+ \rho_\goh, x \rangle}$ by integrating the right-hand side against $e^{- i \langle \mu + \rho_\goh, x \rangle}$ with respect to the normalized Haar measure.  Since $T \cong \tot / 2\pi \Pi(Q_G)^*$ and $|\Pi(Q_G)^*|^{-1} = |\Pi(Q_G)|$, this normalized Haar measure is just $(2 \pi)^{-r} |\Pi(Q_G)|$ times Lebesgue measure on a fundamental domain of $2\pi \Pi(Q_G)^*$, which gives (\ref{eqn:mult-integral}).
\end{proof}

\section{Directions for further research} \label{sec:further}

To conclude, we describe some directions for further research, including potential applications to large-deviations principles in random matrix theory and to tests of entanglement in quantum information theory.

\subsection{Random matrix theory} In recent work, Belinschi, Guionnet and Huang have shown large-deviations principles for the spectral measures of the randomized Horn's and Schur's problems, and have used these to derive related large-deviations principles for Kotska numbers and Littlewood--Richardson coefficients in the limit of large rank \cite{BGH}.  The proofs proceed by relating these objects to Dyson Brownian motion, a stochastic interacting particle system for which the relevant asymptotic behavior is well understood.  The authors note that it might be possible instead to study the large-rank behavior directly from Fourier-analytic expressions for the probability densities in terms of orbital integrals, as derived in \cite{Z, CMZ}.  The main barrier to such an investigation is that it requires knowledge of the asymptotic behavior of the orbital integrals when evaluated at complex arguments, which until recently was not well understood.  However, there is reason to hope that this approach may now be possible due to a recent breakthrough by Novak related to complex asymptotics of orbital integrals \cite{Novak-integrals, Novak-integrals2}.  If such methods can be made to work for the randomized Horn's and Schur's problems, then it should be equally possible to apply them to the quantum marginal problems studied in this paper, which could in turn yield large-deviations principles for Kronecker coefficients.  Another possible approach, which would be of independent interest, is to try to relate quantum marginal problems to Dyson Brownian motion or an analogous stochastic process.

\subsection{Quantum information theory}
If $H$ is a Hilbert space, we write $B(H)$ for the space of bounded linear operators on $H$ and $D(H)$ for the convex subset of $B(H)$ consisting of positive semidefinite self-adjoint operators of trace $1$. In the terminology of quantum information, elements of $D(H)$ are called density matrices. 
Given two finite-dimensional Hilbert spaces $H_1,H_2$, the convex hull of $D(H_1) \otimes D(H_2)$ is clearly 
a subset of $D(H_1\otimes H_2)$, and this inclusion is strict as soon as both $H_1$ and $H_2$ are of
dimension $2$ or greater.
We denote this subset by $\mathrm{Sep}(H_1,H_2)$, and its elements are called separable states, whereas
the elements of $D(H_1\otimes H_2) \setminus \mathrm{Sep}(H_1,H_2)$ are called entangled states. 

It is clear that if $x\in \mathrm{Sep}(H_1,H_2)$, then for any $n \ge 2$, there exists
$x^{(n)}\in D(H_1\otimes \cdots \otimes H_n)$, with $H_2=\ldots =H_n$,
 such that $x^{(n)}$ is invariant under 
any permutation of legs $\{2,\ldots , n\}$ of the tensor product and $\mathrm{id}\otimes \mathrm{id}\otimes \Tr^{\otimes n-2}x^{(n)}=x$.
The quantum de Finetti theorem states that $x\in D(H_1\otimes H_2)$ satisfies such a property for
all $n\ge 2$ if and only if $x$ is actually separable \cite{Stormer, HudsonMoody, Caves}.

Many variants of this theorem exist \cite{KonigRenner, CKMR}, in particular a symmetric one, where the spaces $H_i$ are all equal and $x\in D(H_1\otimes H_2)$ is invariant under exchange of the two legs. In the symmetric setting, the same conclusion holds as in the original quantum de Finetti theorem if in addition one requires $x^{(n)}$ to be invariant under permutations of all $n$ legs.

This approach has been used to describe $\mathrm{Sep}(H_1,H_2)$ as the decreasing limit of
$\mathrm{Ext}_2^n(H_1,H_2)$, the image 
under the map $\mathrm{id} \otimes \mathrm{id} \otimes \Tr^{\otimes n-2}$
of elements of 
$D(H_1\otimes \cdots \otimes H_n)$ that are stable under any permutation of legs $\{2,\ldots , n\}$.
Elements of $\mathrm{Ext}_2^n(H_1,H_2)$ are usually called $n$-extendable density matrices, and therefore
the de Finetti theorem says that being separable is equivalent to being $n$-extendable
for all $n \ge 2$.

If we take an element at random in an orbit of $D(H_1\otimes \cdots \otimes H_n)$
with approximate permutation invariance,
our results allow in principle to compute the distribution of eigenvalues of the pushforward into
$\mathrm{Ext}_2^n(H_1,H_2)$ along the partial trace. 
In the limit of infinite tensor products and with an appropriate choice of an initial orbit in 
$D(H_1\otimes\ldots \otimes H_n)$, we suspect that our calculations could yield either new tests 
of entanglement or at least new natural and exploitable families of probabilities on the set of separable states. 

\section*{Acknowledgements}
The authors would like to thank the three anonymous reviewers for their helpful feedback.  CM would like to thank Jean-Bernard Zuber and Patrick McSwiggen for productive discussions.  The work of BC was supported by JSPS KAKENHI 17K18734 and 17H04823, as well as by Japan--France Integrated Action Program (SAKURA), grant number JPJSBP120203202.  The work of CM was supported by the National Science Foundation under grant number DMS 2103170, and by JST CREST program JPMJCR18T6.

\bibliographystyle{amsplain}

\end{document}